\documentclass[10pt,journal]{IEEEtran}
\usepackage{amsfonts, amsmath, amsthm, amssymb}
\usepackage{mathtools, graphicx, epstopdf, tikz}
\usepackage{multirow}
\usepackage{cite}
\usepackage{subcaption}
\usepackage{atbegshi}
\usepackage{comment}

\IEEEoverridecommandlockouts

\newtheorem{proposition}{Proposition}

\usepackage{algorithm}
\usepackage[noend]{algpseudocode}

\newtheorem{remark}{Remark}

\allowdisplaybreaks
\usepackage{cite, graphicx, textcomp, multirow}
\usepackage{scalefnt}
\begin{document}

	\title{Constellation Design for Nonlinear Unified\\ SWIPT Receiver Channels with Memory\vspace{-2mm}}	
	\author{$\text{Triantafyllos Mavrovoltsos}$, \textit{Graduate Student Member, IEEE}, $\text{Elio Faddoul}$, \textit{Member, IEEE}, $\text{Zulqarnain Bin Ashraf}$, \textit{Graduate Student Member, IEEE}, $\text{Constantinos Psomas}$, \textit{Senior Member, IEEE}, $\text{Besma Smida}$, \textit{Senior Member, IEEE}, and $\text{Ioannis Krikidis}$, \textit{Fellow, IEEE} \thanks{T. Mavrovoltsos, E. Faddoul, and I. Krikidis are with the Department of Electrical and Computer Engineering, University of Cyprus, Nicosia, Cyprus (e-mail: \{tmavro03, efaddo01, krikidis\}@ucy.ac.cy).\\
			Z. B. Ashraf and B. Smida are with the Department of Electrical and Computer Engineering, University of Illinois Chicago, USA (email: \{zbinas2, smida\}@uic.edu).\\
			C. Psomas is with the Department of Computer Science and Engineering, European University Cyprus, Cyprus (email:  c.psomas@euc.ac.cy)} 
		
		\thanks{Part of this work was presented at the IEEE International Communication Conference (ICC), Glasgow, Scotland, May 2026 \cite{Conf}.}\vspace{-8mm}}
	\maketitle
	
	\begin{abstract}	
	Unified receivers (URs) have emerged as a promising architecture for simultaneous wireless information and power transfer (SWIPT), since a common rectifying front-end enables information decoding (ID) and energy harvesting (EH) from the same rectified output. However, rectification is nonlinear due to the diode, while the capacitor introduces memory across symbols, making constellation design over the channel challenging. In this paper, we study constellation design for nonlinear UR-SWIPT channels in both memoryless and memory regimes. First, we propose a tractable unified rectification model that captures both (i) the nonlinear steady-state mapping and (ii) the asymmetric capacitor charging/discharging dynamics under transient operation. To isolate the impact of rectification with memory on ID, we study the information-based design. In this setting, we develop a state-adaptive policy with an algorithmic constellation design that accounts for the rectifier state and shapes the constellation in the observation domain. By approximating the rectifier state distribution, we derive a closed-form average symbol error rate (SER) expression and characterize the rate-reliability (R-R) tradeoff. We then seek constellations that minimize the SER under average transmit power and EH constraints. We address the resulting energy-constrained setting in the memoryless regime using an autoencoder-based framework that embeds the nonlinear rectification model as a differentiable channel block.  Numerical results validate the proposed models, demonstrate the impact of memory on the R-R tradeoff, and show how learned constellations adapt to EH requirements in the rate-energy tradeoff.
	\end{abstract}                                                   
	\vspace{-1mm}

	\vspace{-2mm}
	
	\section{Introduction}

	In the era of sixth-generation (6G) communication networks, a massive number of Internet of Things (IoT) devices is expected to be deployed \cite{6G}. From environmental sensors and biomedical wearables to industrial monitoring nodes, IoT devices will enable a broad range of applications, particularly in smart cities and healthcare\cite{IoT,Survey}. At such large-scale deployments, manual battery charging or replacement becomes impractical and environmentally unsustainable. To address this limitation, simultaneous wireless information and power transfer (SWIPT) has emerged as a promising solution \cite{psomas}. By leveraging energy harvesting (EH) circuits (\textit{i.e.,} rectifiers that convert radio frequency (RF) signals into direct current (DC) output), IoT devices can harvest energy and decode information from the same RF waveform. Nevertheless, efficient SWIPT operation remains challenging due to the fundamental tradeoff between maximizing harvested energy and ensuring reliable information decoding (ID). This has motivated extensive research in (i) energy-efficient receiver architectures \cite{colocated,RuiZhang,Magazine,Diplexer,Roy,Goudeli,Claessens,BASK,Dual}, (ii) accurate mathematical models for the rectification process \cite{NLEH,Boshkovska,alevizos,ayir2023,Thesis,Clerckx1,Morsi,CircuitBased,CircuitBased2,Markov,Zulq,Elena}, and (iii) the design and optimization of SWIPT waveforms \cite{Varasteh,Survey,Morsi,Dual,AutoencoderSWIPT}.
	
	Early SWIPT studies mainly considered receiver architectures in which ID and EH are performed separately~\cite{colocated}. These include separated receivers, with dedicated antennas for ID and EH, and co-located designs that employ power splitting or time switching to allocate the received RF signal between the two functions. Although these architectures have enabled early practical SWIPT implementations, their reliance on active mixers, multiple RF chains, and power splitters increases circuit complexity and energy consumption, motivating more energy-efficient receiver designs. A first step in this direction is the integrated receiver (IR) architecture~\cite{RuiZhang}, which rectifies the RF input and then splits the resulting DC current between baseband detection and EH, thereby partially integrating both tasks within a single circuit. More recently, unified receivers (URs) have attracted significant attention~\cite{Magazine}. By employing a rectifying front-end and avoiding the need for active mixers and RF power splitters, URs substantially reduce hardware complexity and power consumption, making them especially suitable for low-power IoT applications.
	
	Several works have reported notable progress on UR-SWIPT architectures~\cite{Magazine,Diplexer,Roy,Goudeli,Claessens,BASK,Dual}. One line of work considers diplexer-based URs, where a diplexer separates the low- and high-frequency components~\cite{Diplexer}. By eliminating the need for resource splitting techniques, this architecture can achieve improved rate-energy (R-E) tradeoffs. Extending this framework,~\cite{Goudeli} considers multiple-antenna SWIPT with multiple diplexer-based URs, while~\cite{Roy} investigates superimposed chirp waveforms to further enlarge the achievable R-E region. An alternative line of research considers dual-purpose UR architectures based on rectifying circuits~\cite{Claessens,BASK,Dual}, which are also the focus of this work. Based on a half-wave rectifier,~\cite{Claessens,BASK} propose a receiver design in which ID is performed by periodically sampling the steady-state DC voltage, while the same rectified output is used for EH. To further improve the harvested power and the dynamic range of sampled DC levels at the rectifier output, a double half-wave UR architecture is proposed in~\cite{Dual}. A comprehensive overview of various UR architectures is provided in~\cite{Magazine}.

	Accurate mathematical modeling of the rectification process is fundamental to SWIPT system analysis and design~\cite{Survey}. Here, the main challenge is to develop models that are both sufficiently accurate and analytically tractable. Early works often adopted EH models, where harvested power scales linearly with the average RF power~\cite{Varshney,linear}. Nevertheless, this assumption neglects practical rectifier nonlinearities, especially at moderate and high input powers~\cite{alevizos}. To capture this behavior, several nonlinear EH models have been proposed. For instance, parametric power-to-power models, typically fitted from measurements or circuit simulations, map the average RF input power to the average harvested DC output~\cite{NLEH,Boshkovska}. However, because they depend on the average RF power statistics, they provide limited insight for waveform optimization. A more suitable alternative for waveform design is the waveform-to-energy model in~\cite{ayir2023}, which relates harvested power to the distribution of instantaneous RF-envelope power levels. As another example, circuit-based models connect the input waveform with the rectifier output by incorporating device physics and circuit topology~\cite{Thesis,Clerckx1,Morsi,CircuitBased,CircuitBased2}. These models commonly rely on either piecewise-linear (PWL) or exponential current-voltage (I-V) diode characteristics. While PWL models offer tractability by partitioning diode operation into ON/OFF regions~\cite{Thesis}, Shockley-based models better capture the exponential I-V characteristics \cite{Clerckx1,Morsi,CircuitBased,CircuitBased2}.

	Despite these advances, most existing models are memoryless and neglect capacitor charge/discharge dynamics. In UR-SWIPT, such memory effects cannot be ignored, since they directly affect the sampled rectifier output used for ID. Toward more practical modeling, the authors in~\cite{Elena,Zulq} adopt a PWL diode representation and, via an ordinary differential equation (ODE) formulation, investigate memory effects in UR-SWIPT architectures. For a biased amplitude-shift keying (BASK) modulation, they show that the memory induced by the capacitor distorts the rectified waveform and can significantly degrade SWIPT performance. They also demonstrate that sequential detection (\textit{e.g.,} maximum-likelihood (ML) sequence detection) can exploit this memory to improve the symbol error rate (SER), thereby enlarging the achievable rate-reliability (R-R) region. Although ODE-based models accurately capture the transient rectification dynamics, their computational cost makes them impractical for waveform/constellation design. As an alternative, the authors in~\cite{Markov} model circuit memory through a Markov decision process and learn the state transitions from data. However, this approach is developed for separated SWIPT architectures and does not provide an analytical expression for the rectifier output, motivating the need for tractable memory-aware models.
	
Moreover, waveform design plays a central role in SWIPT systems and is closely tied to both the adopted receiver architecture and the rectifier model~\cite{Survey}. Under linear EH models, Gaussian signaling is capacity-achieving and simultaneously maximizes harvested energy~\cite{Varshney,Gaussian}. In this case, the R-E tradeoff is mainly determined by the receiver architecture rather than by the transmit signal distribution~\cite{Survey}. However, under nonlinear EH models this property no longer holds. While Gaussian signaling remains capacity-achieving  over additive white Gaussian noise (AWGN) channels, on-off keying (OOK)-type inputs are known to be favorable for EH because they exploit rectifier's exponential trend~\cite{Varasteh,Morsi}. This mismatch motivates the design of SWIPT waveforms, such as multisine signals with high peak-to-average power ratio (PAPR), that balance ID reliability and harvested energy. In parallel, learning-based waveform design has attracted increasing attention in SWIPT~\cite{MLforSWIPT}. Inspired by deep learning for the physical layer~\cite{HoydisAutoencoder}, autoencoder (AE)-based architectures have been used to jointly optimize transmitter and receiver mappings under nonlinear EH constraints. For example, the authors in~\cite{AutoencoderSWIPT} propose an AE-based SWIPT framework in which the end-to-end system is trained over a nonlinear EH model. The learned constellations were shown to adapt to the EH nonlinearity and agree with prior theoretical insights~\cite{Asymmetric}.
	
At this point, it is important to emphasize that for separated, co-located, and integrated SWIPT receivers, the information channel is typically modeled as linear. However, in dual-purpose UR-SWIPT, front-end rectification makes the effective channel nonlinear, affecting both the harvested energy and the observation used for ID\cite{Claessens}. Consequently, waveform/constellation design becomes more challenging, since the same nonlinearity impacts EH and distorts the received constellations. In this direction,~\cite{BASK} proposes pre- and post-compensation techniques to mitigate distortion and improve ID accuracy, along with a BASK modulation scheme that boosts harvested energy at the expense of ID performance. Furthermore,~\cite{Dual} investigates multitone-based modulation schemes tailored to a double half-wave rectifier, outperforming the single-rectifier baseline. Nevertheless, these approaches use modulation schemes that either yield suboptimal R-E tradeoffs or prioritize a single objective, ignoring memory effects introduced by the capacitor. To address this gap, in this paper, we study constellation design for nonlinear UR-SWIPT channels in both memoryless and memory regimes. The main contributions of this work are summarized as follows:
	\vspace{-1mm}
	
	\begin{itemize}
		
		\item We propose a unified rectification model that yields a tractable discrete-time representation of the UR-SWIPT link. The model captures both (i) the nonlinear steady-state mapping from input amplitude to rectified DC voltage and (ii) the asymmetric exponential charging/discharging dynamics introduced by the rectifier. In addition, it allows us to study memory effects beyond steady-state assumptions, avoiding the need to embed nonlinear ODE solvers within optimization and learning loops.

		\item To isolate the impact of nonlinear rectification and memory on ID, we first study the \emph{information-based} design. We derive algorithmic constellation designs in both the memoryless and memory regimes by enforcing uniformly spaced sampled DC voltage levels at the rectifier output.  For the memory regime, we further propose a state-dependent transmission policy and characterize its average SER performance. By approximating the distribution of the rectifier state, we derive a closed-form expression for the average SER and characterize the impact of memory in the R-R tradeoff.
		
		\item We then formulate constellation design as SER minimization under average transmit power and EH constraints. For the resulting \emph{energy-constrained} setting, we focus on the memoryless regime. Since the resulting optimization problem is nonconvex, we develop an AE-based constellation design framework. By embedding the nonlinear rectification model as a differentiable channel block, we enable the joint optimization of the transmit constellation and detection rule while accounting for the EH constraint during training.
		
		\item Numerical results validate the proposed memory-aware rectification model and the closed-form SER approximation. They also demonstrate the impact of rectifier memory on the R-R tradeoff and show that the optimized \emph{information-based} constellations outperform conventional unipolar PAM. In the \emph{energy-constrained} setting, the learned constellations are used to characterize the R-E tradeoff and reveal that a high-amplitude ``power symbol'' helps satisfy EH requirements while maintaining reliable ID.
		
	\end{itemize}
	
	The remainder of this paper is organized as follows. Section~II presents the considered UR-SWIPT architecture and introduces the unified rectification model. Section~III develops algorithmic constellation designs along with a state-adaptive transmission policy for the \emph{information-based} setting in both the memoryless and memory regimes, and derives the corresponding average SER characterization. Section~IV formulates the \emph{energy-constrained} constellation design problem and presents the AE-based learning framework. Finally, Section~V provides numerical results and Section~VI concludes the paper.
	
	\textit{Notation:} Boldface letters denote vectors; $\mathcal{N}(\mu,\sigma^2)$ denotes real Gaussian random variables with mean $\mu$ and variance $\sigma^2$; $F_{\mathcal{X}}(\cdot)$ and $f_{\mathcal{X}}(\cdot)$ denote the cumulative distribution function (CDF) and probability density function (PDF) of a random variable $\mathcal{X}$, respectively; $\mathbb{E}\{\cdot\}$ is the expectation operator; $Q(\cdot)$ is the Gaussian $Q$-function, $\phi(\cdot)$ and $\Phi(\cdot)$ are the standard normal PDF and CDF, and $\Phi_2(\cdot,\cdot;\rho)$ is the standard bivariate Gaussian CDF with correlation coefficient $\rho$; $I_n(\cdot)$ is the modified Bessel functions of the first kind and order $n$; $W(\cdot)$ denotes the Lambert $\mathcal{W}$ function; and $|\cdot|$ is the absolute value.
	
	\section{System Model}
	\label{sec:sysmodel}
	
	We consider a point-to-point topology in which a single-antenna transmitter communicates with a single-antenna UR featuring a Schottky-based half-wave rectifier, a load resistance $R_L$, and a capacitor $C$, forming a nonlinear (due to the diode) channel with memory (due to the capacitor)\cite{BASK}. The received RF signal $r(t)$ is applied at the rectifier input, and the resulting load voltage $v_L(t)$ across $R_L$ is sampled at symbol boundaries, yielding $v_L^{(n)}$ for ID with additive sampling noise $n_s$, while the same rectified output is used for EH, as shown in Fig. 1.
	
	\vspace{-2mm}
	
	\subsection{Signal Model}
	\label{subsec:setup}\vspace{-0.5mm}
	
	The transmitter communicates $M$ equiprobable messages, where each message carries $\log_2(M)$ bits. Since phase information is lost after rectification, the transmitter employs an amplitude modulation with non-negative amplitude levels.\footnote{In UR-SWIPT designs, front-end rectification makes the receiver observe only the RF envelope, so phase information is largely lost. This naturally motivates non-coherent amplitude/energy-based modulation with non-negative levels, such as BASK or pulse-energy modulation (PEM) schemes~\cite{RuiZhang,BASK}.} Let $A_m \ge 0$ be the amplitude for the $m$-th message, and let $x^{(n)}$ denote the transmitted signal during the $n$-th symbol interval. The received RF signal over $t\in[nT_s,(n+1)T_s)$ follows\vspace{-1mm}
	\begin{equation}
		\begin{aligned}
			r^{(n)}(t) &= x^{(n)}(t) + n_a(t)\\
			&= \sqrt{2P_T} A^{(n)} \cos(2\pi f_c t) + n_a(t),
		\end{aligned}\vspace{-1mm}
	\end{equation}
	where $A^{(n)} \in \{A_0,\ldots,A_{M-1}\}$ is the amplitude of the $n$-th symbol interval, $T_s$ denotes the symbol duration, $P_T$ is the transmit power, $f_c$ is the carrier frequency, and $n_a(t)$ is antenna noise. To satisfy an average transmit power constraint, we have $\mathbb{E}\!\left\{|x^{(n)}|^2\right\}\le P_T$ or equivalently $\frac{1}{M}\sum_{m=0}^{M-1}A_m^2 \le 1$.

	\begin{figure}[t]\centering
		\includegraphics[width=0.9\linewidth]{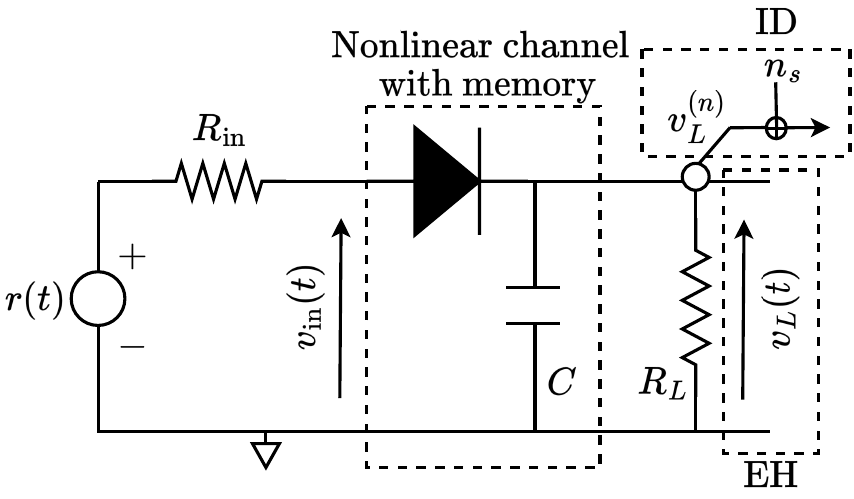}\vspace{-2mm}
		\caption{Considered UR-SWIPT architecture based on a half-wave rectifier\cite{Elena}.}\vspace{-5mm}
	\end{figure}

	\begin{figure*}[t]
		\centering
		\begin{subfigure}{.33\textwidth}
			\centering
			\includegraphics[width=\linewidth]{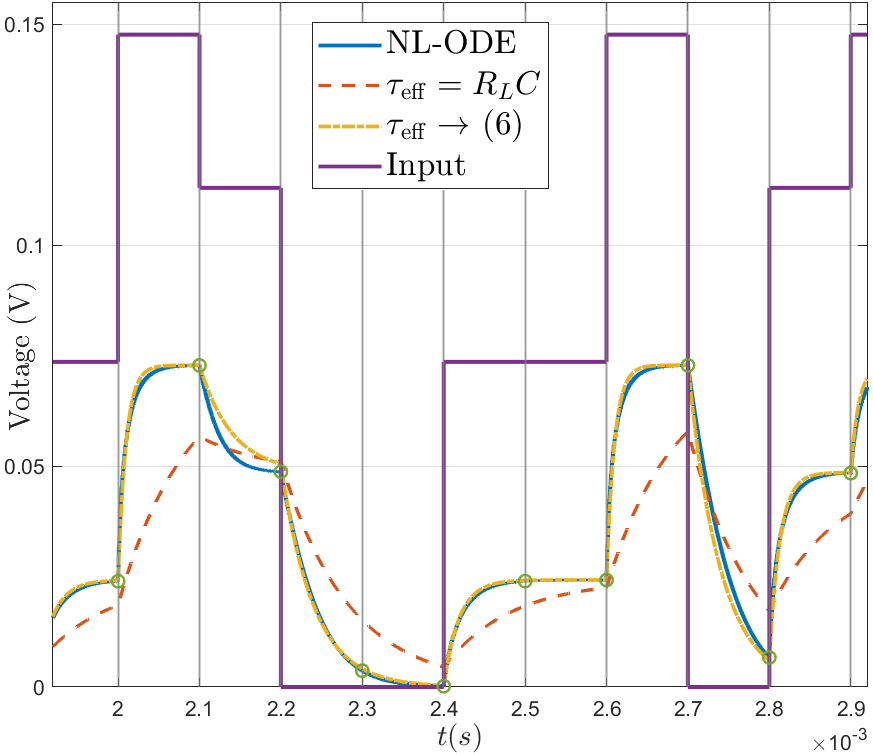}\vspace{-1mm}
			\caption{$T_s=100~\mu\mathrm{s}$, $\alpha_f=0.432, \beta_f=0.485$.}
			\label{fig:MemoryOutput1}
		\end{subfigure}\hfill
		\begin{subfigure}{.33\textwidth}
			\centering
			\includegraphics[width=\linewidth]{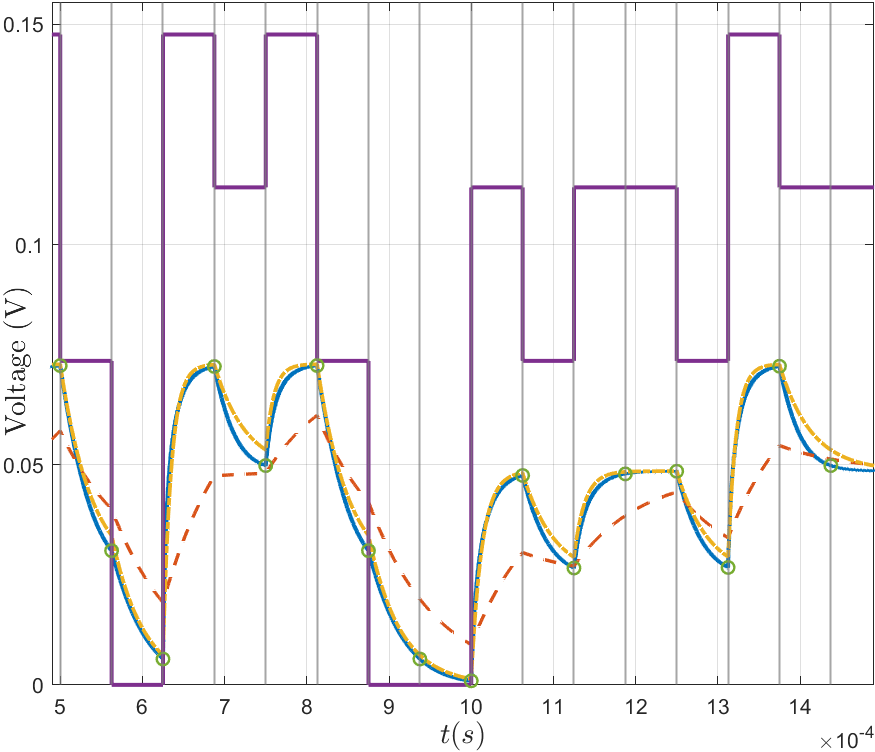}\vspace{-1mm}
			\caption{$T_s=62.5~\mu\mathrm{s}$, $\alpha_f=0.443, \beta_f=0.470$.}
			\label{fig:MemoryOutput2}
		\end{subfigure}\hfill
		\begin{subfigure}{.33\textwidth}
			\centering
			\includegraphics[width=\linewidth]{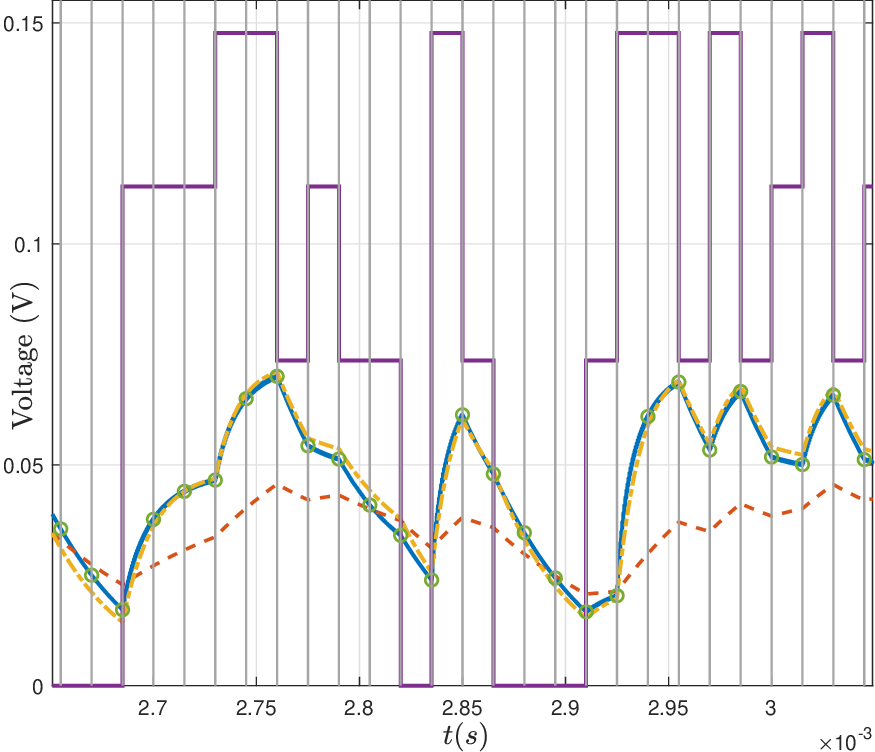}\vspace{-1mm}
			\caption{ $T_s=15~\mu\mathrm{s}$, $\alpha_f=0.490, \beta_f=0.460$.}
			\label{fig:MemoryOutput3}
		\end{subfigure}\hfill\vspace{-1mm}
		\caption{Comparison between the nonlinear ODE model~\eqref{eq:ode_shockley} (solid blue) and the proposed discrete-time approximation~\eqref{eq:dt_model} using (i) a constant time constant $\tau_{\mathrm{eff}}=R_LC$ (dashed red) and (ii) the fitted piecewise $\tau_{\mathrm{eff}}$ in~\eqref{eq:tau_piecewise} (dash-dotted yellow), for three symbol durations. We set $R_L=8.25~\mathrm{k}\Omega$, $C=10~\mathrm{nF}$, $R_{\mathrm{in}}=50~\Omega$, $I_s=5~\mu\mathrm{A}$, $\eta=1.05$, $V_T=25.85~\mathrm{mV}$, and $P_T=-10~\mathrm{dBm}$.}\vspace{-4mm}
		\label{fig:MemoryOutput}
	\end{figure*}

	\subsection{Unified Rectification Model}
	\label{subsec:unified_rectifier}\vspace{-0.5mm}
	
	On the receiver side, the received signal is processed through a Schottky-based half-wave rectifier. Assuming perfect impedance matching, the input voltage at the rectifier is $v_{\mathrm{in}}(t)=r(t)\sqrt{R_{\mathrm{in}}}$, where $R_{\mathrm{in}}$ denotes the input impedance~\cite{Dual}. In the low-power regime, the reverse current becomes negligible and the diode current follows the Shockley equation. By applying Kirchhoff's current law at the diode-load node, we get
	\begin{equation}
		I_s\left(e^{\frac{v_{\mathrm{in}}(t) - v_L(t)}{\eta V_T}} - 1\right) = \frac{v_L(t)}{R_L} + C\frac{\mathrm{d}v_L(t)}{dt},
		\label{eq:ode_shockley}
	\end{equation}
	where $I_s$ is the diode saturation current, $\eta$ is the ideality factor, $V_T$ is the thermal voltage, and $v_L(t)$ is the load voltage. The nonlinear differential equation in (2) provides an exact circuit-level description of the rectifier dynamics. However, repeatedly solving~\eqref{eq:ode_shockley} via nonlinear ODE solvers within an optimization or learning loop is computationally expensive. We therefore develop a closed-form approximation that captures the two mechanisms most relevant to UR-SWIPT design: (i) the nonlinear steady-state mapping induced by the diode and (ii) the transient capacitor charging/discharging behavior that introduces memory when the symbol duration is insufficient for settling. We begin with the steady-state characterization of the considered half-wave rectifier.

	\subsubsection{Steady-state Analysis}
	Let $v_L(t)=\bar v_L+\tilde v_L(t)$, where $\bar v_L$ and $\tilde v_L(t)$ denote the DC and alternating current (AC) components of the load voltage, respectively. Since the rectifier input $v_{\mathrm{in}}(t)$ is periodic with carrier period $T\triangleq 1/f_c$, the ripple component $\tilde v_L(t)$ is also periodic with the same period. Hence, $\tilde v_L(t)=\tilde v_L(t+kT)$ for all $k\in\mathbb{Z}$ and $\frac{1}{T}\int_0^{T}\tilde v_L(t)\,dt=0$. Averaging both sides of \eqref{eq:ode_shockley} over one period yields\vspace{-1mm}
	\begin{equation}
		\frac{1}{T}\int_0^{T} I_s\!\left(e^{\frac{v_{\mathrm{in}}(t)-\bar v_L-\tilde v_L(t)}{\eta V_T}}-1\right)\!\mathrm{d}t
		=
		\frac{\bar v_L}{R_L}.\vspace{-1mm}
	\end{equation}
	 Assuming a sufficiently large capacitor $C$ implies that the ripple $\tilde v_L(t)$ is small and can be neglected. Then, by setting $K=\frac{I_sR_L}{\eta V_T}$, the steady-state voltage corresponding to the $m$-th amplitude is given by~\cite{Morsi}
	\begin{equation}
		\bar v_L(A_m)
		=
		\eta V_T\!\left[W\Big(I_0(\frac{\sqrt{2P_TR_{\mathrm{in}}}\,A_m}{\eta V_T})\,Ke^{K}\Big)-K\right],
		\label{eq:vbar_closed_form}
	\end{equation}
	where the antenna noise $n_a(t)$ is neglected since its contribution to the DC component is negligible~\cite{Dual}. Equation~\eqref{eq:vbar_closed_form} shows that the steady-state mapping is inherently nonlinear, as it results from the exponential characteristic of the Shockley diode equation, while its effective nonlinearity varies with $P_T$, leading to different operating regimes.\vspace{-0.5mm}

	\subsubsection{Memory-aware model}
	\label{subsubsec:memory_model}
	When $T_s$ is not large enough, the load voltage does not reach steady-state within one symbol and the sampled output depends on both the current amplitude and the previous rectifier state~\cite{Elena}. To avoid repeatedly solving~\eqref{eq:ode_shockley}, we approximate the voltage evolution over each symbol interval by a first-order relaxation toward the corresponding steady-state level. Under this approximation and after sampling $v_L(t)$ at $t=nT_s$, the load voltage evolves according to
	\begin{equation}
		v_L^{(n)}\!
		=\!
		v_L^{(n-1)}\!
		+
		\left(\!\bar v_L(A^{(n)})\!-v_L^{(n-1)}\!\right)\!\!
		\left(\!1\!-e^{-T_s/\tau_{\mathrm{eff}}\left(A^{(n)}\right)}\!\right)\!,
		\label{eq:dt_model}\vspace{-1mm}
	\end{equation}
	where $v_L^{(n-1)}\triangleq v_L((n-1)T_s)$ and $v_L^{(n)}\triangleq v_L(nT_s)$ denote the load voltage at the beginning and end of the $n$-th symbol interval, respectively. The effective time constant $\tau_{\mathrm{eff}}(A^{(n)})$ captures the charging/discharging behavior approximated as
	\begin{equation}
		\tau_{\mathrm{eff}}(A^{(n)}) \approx
		\begin{cases}
			\alpha_f \dfrac{R_L C}{1+\dfrac{\bar v_L(A^{(n)})}{\eta V_T}},
			& \bar v_L(A^{(n)}) \ge v_L^{(n-1)}, \\[2mm]
			\qquad \beta_f R_L C,
			& \bar v_L(A^{(n)}) < v_L^{(n-1)},
		\end{cases}
		\label{eq:tau_piecewise}
	\end{equation}
	where $\alpha_f$ and $\beta_f$ are fitting parameters. Further details about the proposed model are provided in the Appendix~\ref{app:memory_model}.
	
	The above approximation reflects the rectifier's asymmetric dynamics. When $\bar v_L(A^{(n)}) \ge v_L^{(n-1)}$, the rectifier operates in the charging regime and the diode conducts. In this case, the diode contributes an additional incremental conductance, which reduces the effective time constant. Moreover, due to the exponential diode I-V characteristic, larger input amplitudes yield stronger diode conduction, leading to faster convergence toward $\bar v_L(A^{(n)})$. This behavior is also captured in~\eqref{eq:tau_piecewise} through the symbol-dependent charging time constant. Conversely, when $\bar v_L(A^{(n)}) < v_L^{(n-1)}$, the diode is mostly reverse-biased and its incremental conductance becomes negligible. The capacitor then discharges mainly through $R_L$, leading to the slower passive time constant, approximately $R_LC$.

	\begin{remark}
		\label{rem:ss_limit}
		From~\eqref{eq:dt_model}, the deviation from the steady-state level \(\bar v_L(A^{(n)})\) evolves as\vspace{-1mm}
		\begin{equation}
			v_L^{(n)}-\bar v_L(A^{(n)})
			=
			\bigl(v_L^{(n-1)}-\bar v_L(A^{(n)})\bigr)\,e^{-T_s/\tau_{\mathrm{eff}}(A^{(n)})}.
			\label{eq:ss_error_decay}\vspace{-1mm}
		\end{equation}
		Hence, as \(T_s/\tau_{\mathrm{eff}}(A^{(n)})\to\infty\) (i.e., \(T_s\gg \tau_{\mathrm{eff}}(A^{(n)})\)), the exponential term vanishes and $v_L^{(n)} \to \bar v_L(A^{(n)}),$ so \eqref{eq:dt_model} reduces to the memoryless steady-state mapping.
	\end{remark}

Fig.~\ref{fig:MemoryOutput} validates the proposed model by comparing it with the numerical solution of the nonlinear ODE in~\eqref{eq:ode_shockley} over the same input sequence for different symbol durations. For $T_s=100~\mu\mathrm{s}$, the load voltage nearly settles to the steady-state value within each interval, so memory effects are weak. In contrast, for $T_s=15~\mu\mathrm{s}$, the capacitor does not fully settle before the next symbol arrives, and the output becomes strongly dependent on the previous rectifier state. As a simple baseline, we also consider the discrete-time model in~\eqref{eq:dt_model} with a constant effective time constant $\tau_{\mathrm{eff}}=R_LC$. Since it ignores the incremental conductance induced by the diode during forward bias, this baseline fails to capture the asymmetric charging/discharging dynamics of the rectifier and leads to mismatch around symbol transitions. In contrast, the proposed piecewise model in~\eqref{eq:tau_piecewise} accounts for both the fast, symbol-dependent charging regime and the slower passive discharging regime, thereby closely matching the nonlinear ODE model in~\eqref{eq:ode_shockley} across all symbol durations.

	\subsection{Information Decoding and Energy Harvesting}
	\label{subsec:ID_EH}\vspace{-0.5mm}
	
	For ID, the UR samples the rectifier output at the end of each symbol interval. 
	From the discrete-time model in~\eqref{eq:dt_model}, the sampled voltage is not determined solely by the current transmit amplitude, but also by the voltage stored on the load capacitor from the previous interval. Hence, the rectifier forms a nonlinear state-dependent channel, and the corresponding noisy observation used for ID can be expressed as\vspace{-1mm}
	\begin{equation}
		y^{(n)}
		= v_L^{(n)} + n_s
		= f_{\mathrm{nl}}\bigl(A^{(n)},v_L^{(n-1)}\bigr) + n_s,
		\label{eq:ID_obs}\vspace{-1mm}
	\end{equation}
	where $f_{\mathrm{nl}}\bigl(\cdot,\cdot\bigr)$ denotes the one-step rectifier mapping induced by~\eqref{eq:dt_model}, and $n_s \sim \mathcal{N}(0,\sigma_s^2)$ is additive noise resulting from the analog-to-digital conversion process~\cite{Elena}. 

	The same rectified output is also used for EH through the load resistance $R_L$. By neglecting the contribution of noise to the harvested energy~\cite{Dual}, the energy collected during the $n$-th symbol interval can be expressed as
	\begin{equation}
		\mathcal{E}_h^{(n)} = \int_{(n-1)T_s}^{nT_s} \frac{v_L^2(t)}{R_L}\, \mathrm{d}t \approx \frac{T_s}{R_L} |v_L^{(n)}|^2.
		\label{eq:Eh_def}
	\end{equation}
where the approximation follows since, for relatively large capacitance, the load voltage varies slowly within one symbol interval~\cite{Elena}. Since the sampled output depends on the transmitted symbol and, under memory, on the previous capacitor state, the rectifier state is a random variable denoted by $\mathcal{V}$. Therefore, the average harvested energy can be written as
	\begin{equation}
		\mathcal{E}_h= \frac{T_s}{R_L}\mathbb{E}_{\mathcal{V}} \left[v^2\right]
		=
		\frac{T_s}{R_L}
		\int_{0}^{\infty} v^2 f_{\mathcal{V}}(v)\,\mathrm{d}v,\label{eq:Ph_avg}
	\end{equation}
	where $f_{\mathcal{V}}(v)$ denotes the PDF of $\mathcal{V}$. In the steady-state regime, the sampled output is $v=\bar v_L(A_m)$. Since the symbols are equiprobable, it reduces to $ \mathcal{E}_h\approx\frac{1}{M}\sum_{m=0}^{M-1}\frac{T_s}{R_L}\left|\bar v_L(A_m)\right|^2$.
	
The expressions in~\eqref{eq:ID_obs} and~\eqref{eq:Ph_avg} show that ID and EH depend on the same sampled rectifier output. Furthermore, in the presence of memory, the average performance of both tasks depends not only on the constellation itself, but also on the resulting state distribution $f_{\mathcal{V}}(v)$ that it induces. In order to isolate the effect of nonlinear rectification and memory on ID performance, we first consider the \emph{information-based} setting, without imposing EH constraints. In what follows, we derive locally optimal constellations through an algorithmic design that applies to both the memoryless and memory regime.
	\vspace{-2mm}
	\section{Information-based Constellation: Algorithmic Design and Insights}\vspace{-0.5mm}

For equiprobable signaling over an AWGN channel, the SER in the asymptotic regime is mainly described by the minimum distance between adjacent observation points. In a linear memoryless channel, the asymptotic SER is minimized by uniformly spacing the transmit amplitudes. However, in the considered UR-SWIPT architecture, uniformly spaced transmit amplitudes do not lead to uniformly spaced sampled voltages. For the \emph{information-based} setting, we shall therefore construct the constellation by enforcing uniform spacing directly at the rectifier output, \textit{i.e.,} in the observation domain.
	
	\vspace{-3mm}
	
	\subsection{Impact of Memory and State-Dependent Constellations}\vspace{-0.5mm}
	
	In the memoryless steady-state regime, a fixed transmit constellation induces a fixed set of sampled voltage levels. In particular, when $T_s\gg \tau_{\mathrm{eff}}$, the rectifier output settles within each symbol interval, and the sampled voltage depends only on the current amplitude through~\eqref{eq:vbar_closed_form}. In this case, the rectifier nonlinearity can be pre-compensated by selecting the transmit amplitudes through inversion of the steady-state mapping~\cite{Conf}. However, when memory effects are non-negligible, the one-step mapping is different and a single static constellation cannot generally be optimal across all possible states. Nevertheless, for any fixed state $V_0$, a locally optimal \emph{information-based} constellation can be obtained by uniformly spacing the corresponding sampled voltage levels.

	As a first step, we need to identify the lowest sampled voltage that can be reached from an initial state $V_0$ within one symbol interval. Obviously, this minimum is obtained by transmitting the zero-amplitude symbol. In this case, the diode is mostly reverse-biased and the capacitor discharges through $R_L$, so the load voltage decays approximately as $v_L(t)\approx V_0 e^{-t/(\beta_f R_L C)}$. Therefore, the minimum sampled voltage reachable from $V_0$ over one symbol is given by\vspace{-1mm}
	\begin{equation}
		V_{\min}(V_0)= V_0 e^{-T_s/(\beta_f R_L C)}.
		\label{eq:Vmin}\vspace{-1mm}
	\end{equation}
For a spacing $d$, the uniformly spaced target voltages are\vspace{-1mm}
\begin{equation}
	v_L^m(d,V_0)=V_{\min}(V_0)+m d.
	\label{eq:target_voltages}\vspace{-1mm}
\end{equation}
To obtain the corresponding transmit amplitudes, we invert the one-step mapping with respect to the amplitude. For a fixed rectifier state $V_0$, this gives\vspace{-1mm}
\begin{equation}
	A_m(d,V_0)
	=
	f_{\mathrm{nl}}^{-1}\bigl(v_L^m(d,V_0);V_0\bigr),
	\label{eq:Am_inverse}\vspace{-1mm}
\end{equation}
where $f_{\mathrm{nl}}^{-1}(\cdot;V_0)$ denotes the inverse of the one-step mapping for fixed state $V_0$. In general, this inverse is not available in closed form, but $A_m(d,V_0)$ can be computed numerically using standard root-finding methods.\footnote{The rectification function must be invertible over the non-negative amplitude range of interest. However, the AE-based framework in Section~IV remains applicable even for non-monotonic rectifiers (\textit{e.g.,}~\cite{Me}) without requiring an explicit inverse.}  To satisfy the transmit power constraint, the maximum feasible spacing is defined as
\begin{equation}
	d_{\max}(V_0)
	=
	\max_{d\geq 0}
	\left\{
	d:
	\frac{1}{M}
	\sum_{m=0}^{M-1}
	\big[
	A_m(d,V_0)
	\big]^2
	\leq 1
	\right\}.
	\label{eq:dmax_definition}
\end{equation}
Since the average transmit power is monotonic in $d$, $d_{\max}(V_0)$ can be found by bisection, and the largest feasible spacing is obtained by increasing $d$ until the average power constraint is met. The procedure is summarized in Algorithm~\ref{alg:dmax_mem}. Once $d_{\max}(V_0)$ is determined, the resulting constellation $\mathbf{A}(V_0)=\{A_m(V_0)\}_{m=0}^{M-1}$ is locally optimal for a given rectifier state $V_0$. Notably, when $T_s$ is sufficiently large and the transient term vanishes and $f_{\mathrm{nl}}(A,V_0)\approx \bar v_L(A)$ becomes independent of $V_0$. In this case, Algorithm~\ref{alg:dmax_mem} reduces to the memoryless design.

	\begin{algorithm}[t]
		\caption{State-dependent bisection for \emph{information-based} constellation design}
		\label{alg:dmax_mem}
		\begin{algorithmic}[1]
			\State \textbf{Input:}  $M$, $V_0$, $T_s$, $R_L$, $C$, $\alpha_f$, $\beta_f$, tolerance $\epsilon$; mappings $f_{\mathrm{nl}}(\cdot,V_0)$ and $f_{\mathrm{nl}}^{-1}(\cdot;V_0)$
			\State \textbf{Output:}  $\mathbf{A}(V_0)=\{A_m(V_0)\}_{m=0}^{M-1}$
			
			\State $V_{\min}\gets V_0e^{\left(-T_s/\beta_f R_L C\right)}$
			\State $d_\ell\gets 0$, $d_u\gets \epsilon$
			
			\Repeat
			\For{$m=0,\ldots,M-1$}
			\State $A_m\gets f_{\mathrm{nl}}^{-1}(V_{\min}+m d_u;V_0)$
			\EndFor
			\State $P_u\gets \frac{1}{M}\sum_{m=0}^{M-1}A_m^2$
			\If{$P_u\leq 1$}
			\State $d_\ell\gets d_u$, $d_u\gets 2d_u$
			\EndIf
			\Until{$P_u>1$}
			
			\While{$d_u-d_\ell>\epsilon$}
			\State $d\gets (d_\ell+d_u)/2$
			\For{$m=0,\ldots,M-1$}
			\State $A_m\gets f_{\mathrm{nl}}^{-1}(V_{\min}+m d;V_0)$
			\EndFor
			\State $P(d)\gets M^{-1}\sum_{m=0}^{M-1}A_m^2$
			\If{$P(d)>1$}
			\State $d_u\gets d$
			\Else
			\State $d_\ell\gets d$
			\EndIf
			\EndWhile
			
			\State $d_{\max}(V_0)\gets d_\ell$
			\For{$m=0,\ldots,M-1$}
			\State $A_m(V_0)\gets f_{\mathrm{nl}}^{-1}(V_{\min}+m d_{\max}(V_0);V_0)$
			\EndFor
			\State \Return $\mathbf{A}(V_0)$
		\end{algorithmic}
	\end{algorithm}

	Fig.~3(a) shows the resulting transmit constellations $\mathbf{A}(V_0)$ across different rectifier states. As $V_0$ increases, the constellation changes because each amplitude may induce either charging or discharging, depending on its steady-state voltage $\bar v_L(A_m)$. For small rectifier states, \textit{e.g.} $V_0=0.01~\mathrm{V}$, all nonzero amplitudes satisfy $\bar v_L(A_m)>V_0$ and therefore operate in the charging regime. As $V_0$ increases, some of the symbols cross the transition point $\bar v_L(A_m)=V_0$ and move into the discharging regime. For example, when $V_0\gtrsim 0.033~\mathrm{V}$, the lowest nonzero amplitude enters the discharging regime, whereas for $V_0\gtrsim 0.06~\mathrm{V}$, only the largest amplitude remains in the charging regime. This trend follows from the piecewise time constant in~\eqref{eq:tau_piecewise}, which separates the forward-biased and reverse-biased diode regimes. In addition, Fig.~3(b) shows the corresponding effective constellations in the observation domain. As \(V_0\) increases, the lowest reachable voltage \(V_{\min}(V_0)\) also increases, and the sampled voltage levels shift upward. At the same time, this shift reduces the maximum feasible spacing $d_{\max}(V_0)$ and the sampled voltage levels become more compressed. Additionally, for larger rectifier states, most symbols mainly drive the diode into the reverse-biased regime. In this regime, the capacitor discharges more slowly, which limits the voltage change within one symbol interval and further compresses the effective constellation.
	
	\vspace{-2mm}
	\begin{remark} For large rectifier states, where all symbols drive the diode in the reverse-biased regime, the effective time constant in~\eqref{eq:tau_piecewise} is simply described by $\tau_{\mathrm{eff}}=\beta_f R_LC$. Then, the one-step mapping reduces to\vspace{-1mm}
		\begin{equation}
			v_L^{(n)}
			=
			e^{-\frac{T_s}{\beta_f R_LC}}v_L^{(n-1)}
			+
			\left(1-e^{-\frac{T_s}{\beta_f R_LC}}\right)\bar v_L(A_m).\vspace{-1mm}
		\end{equation}
		and the sampled voltages become an affine transformation of the steady-state mapping $\bar v_L(A_m)$. Since the affine mapping is increasing, this case gives the same transmit amplitudes as the steady-state design, with only the output spacing scaled.
	\end{remark}\vspace{-2mm}

	\begin{figure}[t]
	\centering
	\begin{subfigure}{.24\textwidth}
		\centering
		\includegraphics[width=\linewidth]{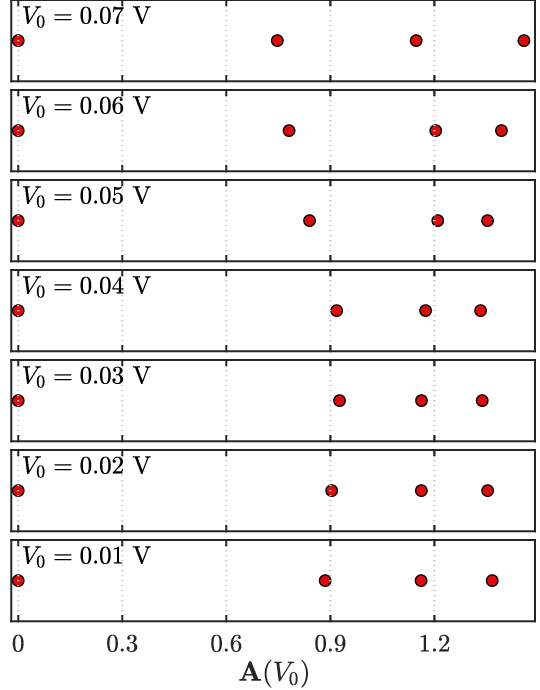}\vspace{-1mm}
		\caption{Transmit constellations.}\vspace{-1mm}
	\end{subfigure}%
	\begin{subfigure}{.24\textwidth}
		\centering
		\includegraphics[width=\linewidth]{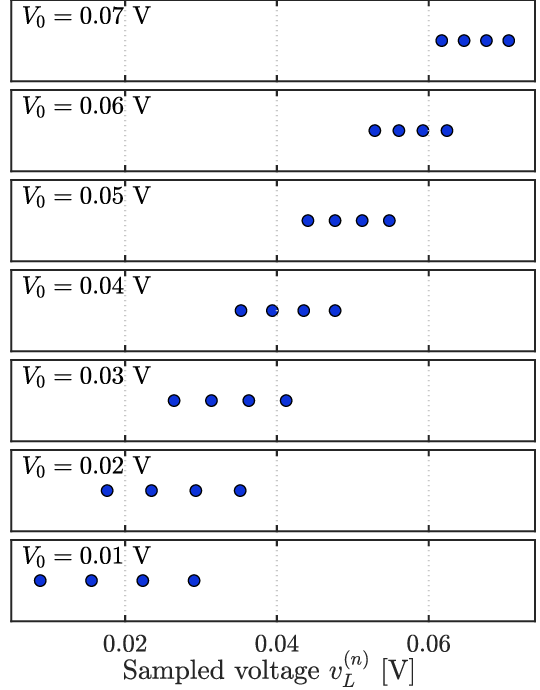}\vspace{-1mm}
		\caption{Effective constellations.}\vspace{-1mm}
	\end{subfigure}
	\caption{Locally optimal \emph{information-based} constellation $\mathbf{A}(V_0)$, along with the corresponding sampled voltages, with $V_0=\{0.01,0.02,\ldots,0.07\}$. We set $P_T=-10~\mathrm{dBm}$, $T_s=10~\mu\mathrm{s}$, and $\alpha_f=0.44$, $\beta_f=0.48$.}\vspace{-4mm}
\end{figure}

Overall, rectifier memory makes the effective constellation depend on the capacitor voltage at the beginning of each symbol interval. As a result, the effective constellation spacing and the instantaneous SER vary with the rectifier state. This motivates a state-adaptive transmission policy, in which the transmit constellation is adapted to the current rectifier state. In what follows, we characterize the corresponding state distribution and derive the average SER of the proposed scheme.

		\vspace{-3mm}

\subsection{Rectifier State Distribution and SER Performance}
\label{subsec:state_distribution_ser}

We now consider a state-adaptive \emph{information-only} scheme, where the transmitter is assumed to know the rectifier state $V_0$. At the beginning of each symbol interval, the transmitter constructs the locally optimal $M$-ary constellation $\mathbf{A}(V_0)$ and selects one of its $M$ amplitudes with equal probability.\footnote{If instantaneous state information is unavailable at the transmitter, one could instead design a fixed constellation that is robust over the typical operating range of $V_0$, which we leave for future work.} For a fixed state $V_0$, Algorithm~\ref{alg:dmax_mem} produces uniformly spaced sampled voltage levels with spacing $d_{\max}(V_0)$. 	Therefore, under minimum-distance detection, the conditional SER is given by~\cite{Proakis}\vspace{-1mm}
\begin{equation}
	\mathrm{SER}\bigl(\mathbf{A}(V_0); V_0\bigr)
	=
	\frac{2(M-1)}{M}
	Q\!\left(\frac{d_{\max}(V_0)}{2\sigma_s}\right).
	\label{eq:ser_cond_V0}\vspace{-1mm}
\end{equation}

As already mentioned, the rectifier state is a random variable, denoted by $\mathcal{V}$. Each state realization $v$ induces a different feasible spacing $d_{\max}(v)$ and, consequently, a different conditional SER. Hence, the average SER of the state-adaptive scheme is then obtained by averaging~\eqref{eq:ser_cond_V0} over the distribution of $\mathcal{V}$, yielding\vspace{-1mm}
\begin{equation}
	\overline{\mathrm{SER}}\! \triangleq\! \mathbb{E}_{\mathcal{V}}\! \!\left[ \mathrm{SER}\bigl(\mathbf{A}(\mathcal{V}) ;\! \mathcal{V}\bigr) \right]\!\! =\!\! \int_{-\infty}^{\infty}\!\!\!\!\!\! \mathrm{SER}\bigl(\mathbf{A}(v);\! v\bigr)\! f_{\mathcal{V}}(v) \mathrm{d}v, \label{eq:avg_SER_def}\vspace{-1mm}
\end{equation}
In general, $f_{\mathcal{V}}(v)$ is not available in closed-form since the state evolves through the nonlinear mapping and depends on both the constellation and symbol probabilities. However, it can be estimated numerically by simulating~\eqref{eq:dt_model} over long sequences.

		\begin{figure}[t]
		\centering
		\begin{subfigure}{.242\textwidth}
			\centering
			\includegraphics[width=\linewidth]{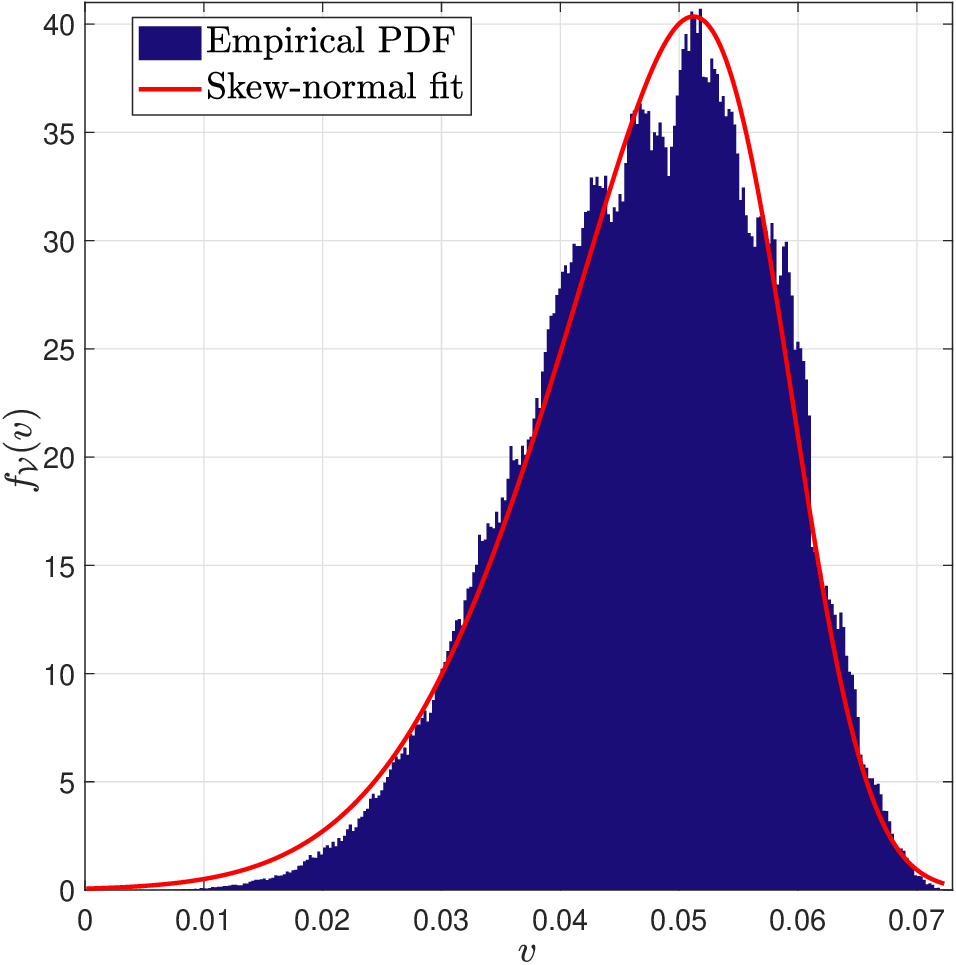}\vspace{-1mm}
			\caption{$T_s=10~\mu\mathrm{s}$.}
		\end{subfigure}\hfill
		\begin{subfigure}{.242\textwidth}
			\centering
			\includegraphics[width=\linewidth]{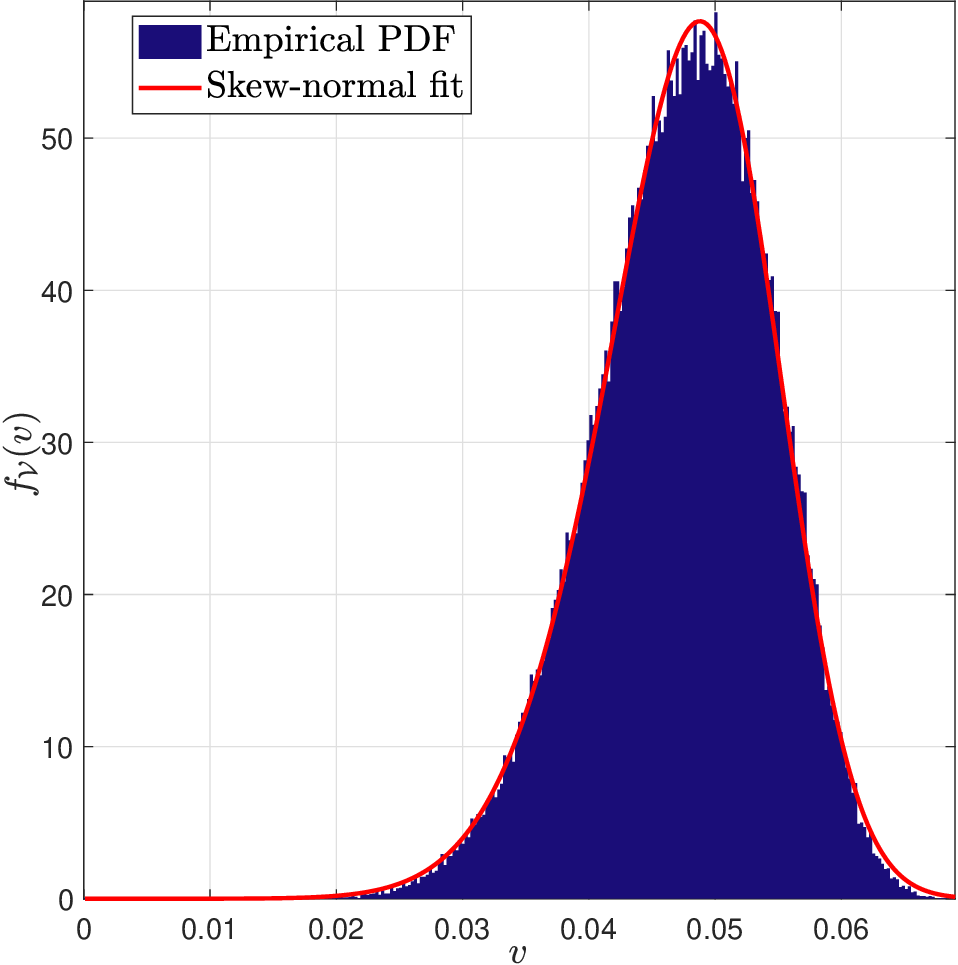}\vspace{-1mm}
			\caption{$T_s=5~\mu\mathrm{s}$.}
		\end{subfigure}

		\begin{subfigure}{.242\textwidth}
			\centering
			\includegraphics[width=\linewidth]{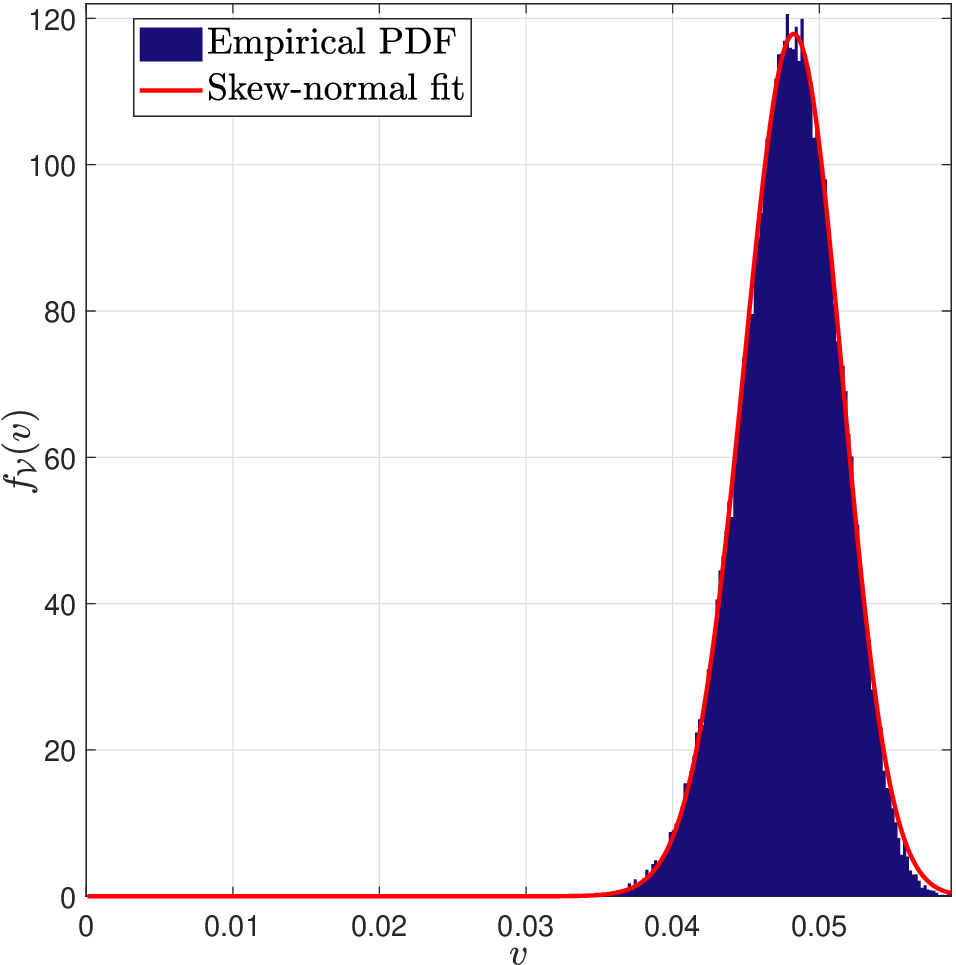}\vspace{-1mm}
			\caption{$T_s=1~\mu\mathrm{s}$.}
		\end{subfigure}\hfill
		\begin{subfigure}{.242\textwidth}
			\centering
			\includegraphics[width=\linewidth]{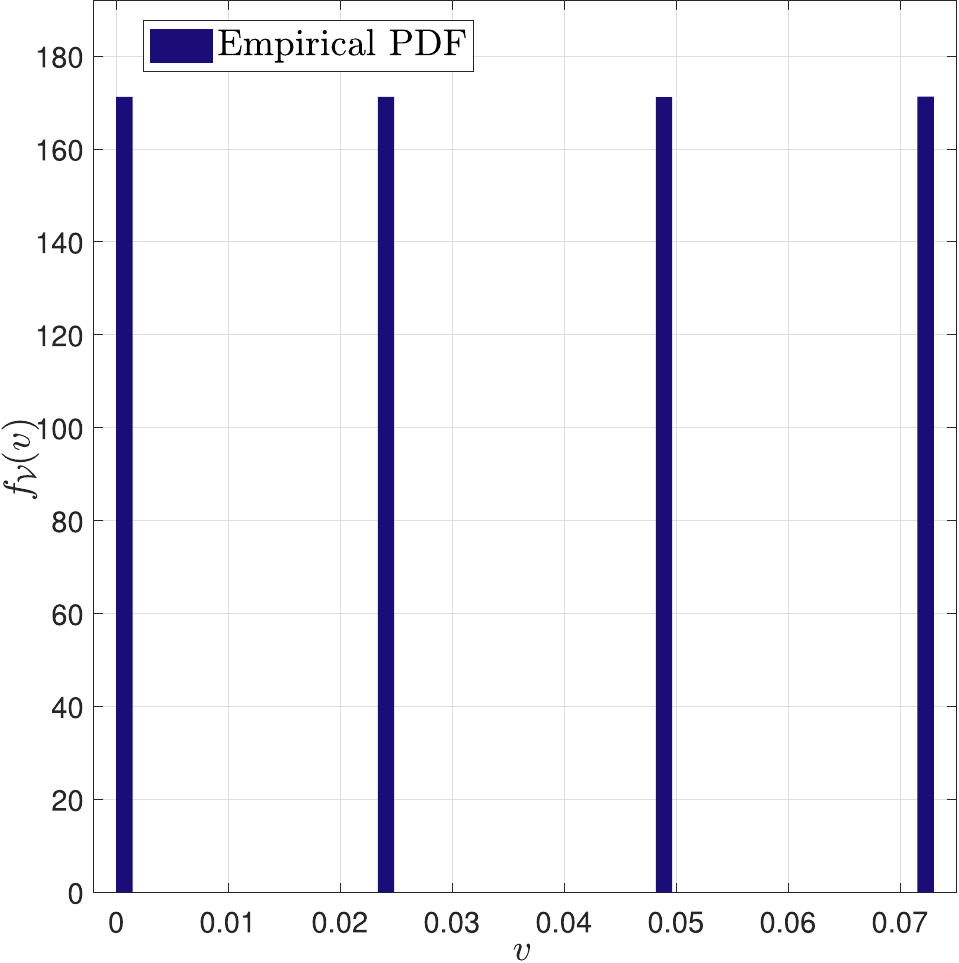}\vspace{-1mm}
			\caption{$T_s=1\mathrm{ms}$.} 
		\end{subfigure}\vspace{-1mm}
\caption{Empirical distribution of the rectifier state under state-adaptive \emph{information-based} signaling, for $P_T=-10~\mathrm{dBm}$ and (a) $T_s=10~\mu\mathrm{s}$, (b) $T_s=5~\mu\mathrm{s}$, (c) $T_s=1~\mu\mathrm{s}$, and (d) $T_s=1~\mathrm{ms}$.} %
		\label{fig:V0_hist}\vspace{-4mm}
	\end{figure}
	
			\begin{figure*}\centering
		\includegraphics[width=0.93\textwidth]{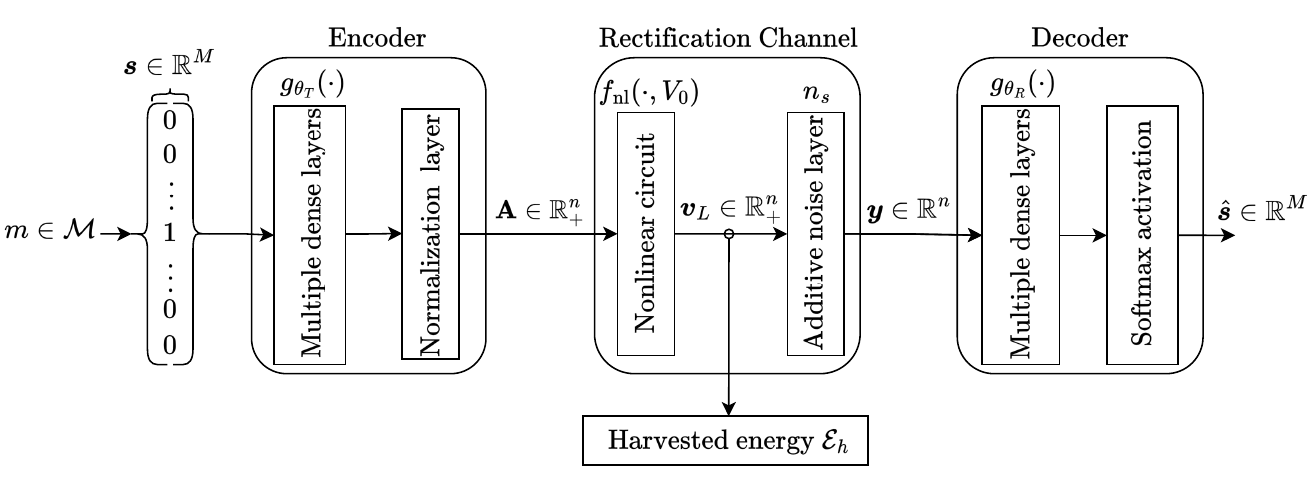}\vspace{-2mm}
		\caption{AE structure for UR-SWIPT.}\vspace{-5mm}
	\end{figure*}	
	
	Fig.~\ref{fig:V0_hist} shows the empirical distribution of the rectifier state for different symbol durations. For $T_s\in{1,5,10}~\mu\mathrm{s}$, the state is mainly concentrated around $v\approx 0.05~\mathrm{V}$, while lower-voltage states are less probable. This is a direct consequence of the asymmetric charging and discharging dynamics of the rectifier. In particular, moderate- and high-amplitude symbols can quickly charge the capacitor, while reaching a lower state requires several consecutive low-amplitude symbols, as also reflected in the voltage evolution in Fig.~\ref{fig:MemoryOutput}. As $T_s$ decreases, the capacitor has less time to change within each symbol interval, and the state distribution becomes narrower. For comparison, Fig.~\ref{fig:V0_hist}(d) shows the memoryless case with $T_s=1~\mathrm{ms}$. In this case, the rectifier reaches steady-state within each symbol interval, and the state reduces to a finite set of equiprobable voltage levels corresponding to the transmitted amplitudes.

	To obtain a tractable characterization of the rectifier state statistics, we approximate the empirical state distribution using a skew-normal model. Specifically, the PDF of $\mathcal{V}$ is fitted as
	\begin{equation}
		f_{\mathcal{V}}(v)
		\approx
		\frac{2}{\omega}
		\phi\!\left(\frac{v-\xi}{\omega}\right)
		\Phi\!\left(\alpha\frac{v-\xi}{\omega}\right),
		\label{eq:skew_normal_pdf}
	\end{equation}
	where $\xi$, $\omega$, and $\alpha$ denote the fitted location, scale, and shape parameters, respectively. As shown in Fig.~\ref{fig:V0_hist}, the skew-normal model provides a good fit in the memory regime of interest. 
	Substituting~\eqref{eq:ser_cond_V0} and~\eqref{eq:skew_normal_pdf} into~\eqref{eq:avg_SER_def}, the average SER under rectifier memory is approximated as
	\begin{equation}
		\begin{aligned}
			\overline{\mathrm{SER}}
			\approx
			\int_{-\infty}^{\infty}
			&\frac{2(M-1)}{M}
			Q\!\left(\frac{d_{\max}(v)}{2\sigma_s}\right) \\
			&\times
			\frac{2}{\omega}
			\phi\!\left(\frac{v-\xi}{\omega}\right)
			\Phi\!\left(\alpha\frac{v-\xi}{\omega}\right)
			\mathrm{d}v .
		\end{aligned}
		\label{eq:avg_ser_skew_normal}
	\end{equation}
	However, the integral in~\eqref{eq:avg_ser_skew_normal} is still not available in closed form since the optimized spacing $d_{\max}(v)$ depends on the rectifier state and has no explicit analytical expression. For analytical tractability, we approximate $d_{\max}(v)$ over the high-probability operating range of $\mathcal{V}$. In the considered regime, $d_{\max}(v)$ can be approximated as an affine function, yielding
	\begin{equation}
		d_{\max}(v)
		\approx
		d_0+d_1v,
		\label{eq:dmax_affine_approx}
	\end{equation}
	where $d_1=\frac{d_{\max}(v_{\max})-d_{\max}(v_{\min})}{v_{\max}-v_{\min}},$ $d_0=d_{\max}(v_{\min})-d_1v_{\min},$ with $v_{\min}$ and $v_{\max}$ denoting the lower and upper limits of the high-probability operating range. Under this approximation, the average SER admits a closed-form expression given in the following proposition.
	
	\begin{proposition}
		\label{prop:closed_form_ser_memory}
		Under the skew-normal state model in~\eqref{eq:skew_normal_pdf} and the affine approximation in~\eqref{eq:dmax_affine_approx}, the average SER of the state-adaptive information-based constellation is approximated as
		\begin{equation} 
		\overline{\mathrm{SER}}\! \approx \!\frac{4(M \! -1)}{M} \Phi_2 \!\left(\! 0, \!\frac{-A}{\sqrt{1\!+\! B^2}}; \!\frac{-\alpha B} {\sqrt{(1\!+\alpha^2)(1\!+\! B^2)}}\!\! \right)\!\!,
		\label{eq:avg_ser_closed_form_expanded} \end{equation}
		where $
			A=
			\frac{d_0+d_1\xi}{2\sigma_s},$ and $B=
			\frac{d_1\omega}{2\sigma_s}.$
	\end{proposition}
	
	\begin{proof}
		See Appendix~\ref{app:ser_skew_normal}.
	\end{proof}
	
Proposition~1 provides a compact approximation for the average SER and shows how it is affected by the rectifier state distribution and the state-dependent spacing $d_{\max}(v)$. In the considered regime, $d_{\max}(v)$ generally decreases with $v$, so larger rectifier states result in smaller voltage spacing and higher SER. If $d_1=0$, the spacing is fixed and the result reduces to the standard AWGN SER.

The above \emph{information-based} setting characterizes how rectifier memory affects ID. However, in UR-SWIPT, the same rectified output is also used for EH. Hence, the constellation must not only provide reliable detection, but also generate sufficient harvested energy. In the next section, we formalize this requirement through a minimum EH constraint and develop an AE-based framework to address the resulting joint ID-EH constellation design.

	\section{Energy-constrained Constellations: An AE-based Framework}
	
	We now study how a minimum EH requirement affects the \emph{energy-constrained} constellation structure. For demonstration, we adopt the memoryless steady-state regime, where the AE learns the constellation over the nonlinear mapping $\bar v_L(\cdot)$ given in~\eqref{eq:vbar_closed_form}. To keep this effect isolated, we consider one channel use per message.\footnote{The same one-step learning framework can also be applied under memory by replacing $\bar v_L(\cdot)$ with the state-dependent mapping $f_{\mathrm{nl}}(\cdot,V_0)$ for a given rectifier state.} Sequence-level learning over the full state-dependent mapping is left for future work. We first formulate the corresponding constrained SER minimization problem and then present the AE-based implementation.

\vspace{-2mm}
\subsection{Problem Formulation}\vspace{-0.5mm}
 Let $\mathrm{SER}\left(\mathbf{A};\mathcal{D}(\cdot)\right)$ denote the SER achieved by the constellation $\mathbf{A}=\{A_0,A_1,\ldots,A_{M-1}\}$, with $A_m\ge0$, under decoding rule $\mathcal{D}(\cdot)$. 
We formulate the design problem as 
\begin{subequations} 
	\begin{align} \min_{\mathbf{A},\mathcal{D(\cdot)}} \quad & \mathrm{SER}\!\left(\mathbf{A};\mathcal{D(\cdot)}\right) \label{eq:opt_obj} \\ \text{s.t.}\quad & \frac{1}{M}\sum_{m=0}^{M-1}A_m^2 \le 1 \label{eq:opt_power} \\[1mm] & \mathcal{E}_h(\mathbf{A}) \ge \mathcal{E}_t, \label{eq:opt_eh} \end{align} \label{eq:opt_prob} 
\end{subequations}
$\!\!\!\!\!$ where $\mathcal{E}_t$ denotes the required harvested energy. The above problem is nonconvex since both the SER and the EH constraint depend on the constellation through the nonlinear rectification mapping $\bar v_L(\cdot)$. In particular, the sampled voltage levels are nonlinear functions of the transmit amplitudes, while the EH constraint depends on their squared rectified values. Therefore, a closed-form solution is generally not available. To tackle this, we model the end-to-end UR-SWIPT link as an AE. In the following, we describe the proposed AE architecture and its training procedure.

	\vspace{-3mm}
	
	\subsection{Autoencoder Structure for UR-SWIPT}\vspace{-0.5mm}
	
To implement the proposed \emph{energy-constrained} design, we represent the end-to-end UR-SWIPT link as an AE, as shown in Fig.~5. The encoder models the transmitter, the rectification channel represents the nonlinear rectification, and the decoder performs ID from the noisy sampled voltages. The transmitter communicates one of $M$ equiprobable messages, $m\in\mathcal{M}=\{0,1,\ldots,M-1\}$. Each message is represented by a one-hot vector $\boldsymbol{s}\in\mathbb{R}^M$, whose only nonzero entry is at the position corresponding to $m$. After the nonlinear rectification, the receiver (or decoder) jointly harvests energy and decodes information from the rectified signals.

The encoder maps the one-hot vector  $\boldsymbol{s}$ into a non-negative amplitude vector $\mathbf{A}\in\mathbb{R}_+^n$, where $n$ denotes the channel uses. The mapping from the set of messages $\mathcal{M}$ to the signal space $\mathcal{X}^n$ is defined by a parametric function $g_{\theta_T}(\cdot): \mathcal{M} \rightarrow \mathbb{R}^n$, where $\theta_T$ represents the set of encoder parameters, including the weights and biases across the layers of the encoder. To satisfy the average transmit power constraint, the encoder employs a normalization layer as its final stage. Since the encoded messages $\mathbf{A}$ are constrained to non-negative amplitudes, this layer outputs the normalized absolute value of its input. 

The encoded amplitude vector passes through the nonlinear rectification channel $f_{\mathrm{nl}}(\cdot,V_0)$, as shown in Fig.~5. In the considered steady-state regime, this mapping reduces to $\bar v_L(\cdot)$ and becomes independent of $V_0$.\footnote{In this work, the function \(f_{\mathrm{nl}}(\cdot,V_0)\) is known. If it is obtained from measurements, an interpolated differentiable approximation can be used for backpropagation. This further motivates the learning-based approach, as it enables constellation design tailored to the actual rectifier characteristics.} Here, the Lambert \(\mathcal{W}\) term in~\eqref{eq:vbar_closed_form} is evaluated using a fixed number of differentiable Newton iterations, which enables backpropagation through the rectifier model. The resulting sampled voltage vector $\boldsymbol{v}_L\in\mathbb{R}_+^n$ is used to compute the harvested energy $\mathcal{E}_h$, while its noisy version,
$\boldsymbol{y}=\boldsymbol{v}_L+\boldsymbol{n}_s$, is passed to the decoder for ID. The decoder, parameterized by $g_{\theta_R}(\cdot)$, maps $\boldsymbol{y}\in\mathbb{R}^n$ to a probability vector $\hat{\boldsymbol{s}}\in\mathbb{R}^M$, where $\theta_R$ denotes the decoder weights and biases. The detected message is obtained by selecting the index with the largest output probability, \textit{i.e.,} $\hat{m}=\arg\max_i \hat{s}_i$. The AE's layout is summarized in Table~I.

	\begin{table}[!t]
		\centering
		\caption{\\ \small Layout of the UR-SWIPT autoencoder}\small
		\begin{tabular}{l c}
			\hline\hline
			\textbf{Layers} & \textbf{Output shape} \\
			\hline
			\textbf{Encoder:} & \\
			\quad Input (One-hot) & $\mathbb{R}^M$ \\
			\quad Dense (Leaky ReLU) & $\mathbb{R}^M$ \\
			\quad Dense (Linear) & $\mathbb{R}^n$ \\
			\quad Normalization & $\mathbb{R}^n$ \\ 
			\hline
			\textbf{Channel:} & \\
			\quad Non-trainable function & $\mathbb{R}^n$ \\
			\quad AWGN & $\mathbb{R}^n$ \\
			\hline
			\textbf{Decoder:} & \\
			\quad Dense (Leaky ReLU) & $\mathbb{R}^M$ \\
			\quad Dense (Softmax) & $\mathbb{R}^M$ \\
			\hline\hline
		\end{tabular}\vspace{-5mm}
	\end{table}
	
\vspace{-3mm}

	\subsection{Training the Autoencoder}
	
The encoder and decoder parameters are trained to improve ID reliability while satisfying the EH requirement. To this end, we model the information loss as the cross-entropy function between the transmitted one-hot vector $\boldsymbol{s}$ and the decoder output probability vector $\hat{\boldsymbol{s}}$, given by
\begin{equation}
	L_{\mathrm{CE}}(\boldsymbol{s},\hat{\boldsymbol{s}})=-\sum_{i=1}^{M}s_i\log \hat{s}_i,
\end{equation}
where $s_i$ and $\hat{s}_i$ denote the $i$-th entries of $\boldsymbol{s}$ and $\hat{\boldsymbol{s}}$, respectively. To incorporate the EH constraint in~\eqref{eq:opt_eh}, we add a one-sided penalty that is active only when the average harvested energy is below the target $\mathcal{E}_t$. The resulting mini-batch loss is 
\begin{equation}
	L(\theta_T,\theta_R)\!
	=\!
	\frac{1}{|\mathcal{B}_m|}
	\sum_{k\in\mathcal{B}_m}\!\!\!
	L_{\mathrm{CE}}\bigl(\boldsymbol{s}^{(k)},\hat{\boldsymbol{s}}^{(k)}\bigr)
	+
	\lambda\left[\mathcal{E}_t-\bar{\mathcal{E}}_h\right]_+^2,
	\label{eq:ae_loss}
\end{equation}
where $[x]_+=\max\{x,0\}$, $\mathcal{B}_m$ denotes the mini-batch, $\lambda\geq 0$ controls the weight of the EH penalty, and
\begin{equation}
	\bar{\mathcal{E}}_h
	=
	\frac{1}{|\mathcal{B}_m|}
	\sum_{k\in\mathcal{B}_m}
	\mathcal{E}_h\bigl(\mathbf{A}^{(k)}\bigr)
	\label{eq:batch_energy}
\end{equation}
is the average harvested energy over the mini-batch. For the $k$-th training sample, $\boldsymbol{s}^{(k)}$ is the one-hot vector, $\hat{\boldsymbol{s}}^{(k)}$ is the decoder output, and $\mathbf{A}^{(k)}$ is the corresponding encoded amplitude vector.

In practice, the AE is trained by minimizing~\eqref{eq:ae_loss} over randomly generated message samples using mini-batch stochastic optimization with the Adam optimizer~\cite{Adam}. When $\lambda=0$, the EH penalty is removed and the AE learns an \emph{information-based} constellation that prioritizes ID reliability. For $\lambda>0$, the penalty encourages the learned constellation to satisfy the target harvested energy $\mathcal{E}_t$ while preserving ID reliability. The value of $\lambda$ is selected empirically so that the EH constraint is satisfied without overshadowing the cross-entropy term. The training signal-to-noise ratio (SNR) is also an important design parameter and is selected according to the transmit power $P_T$, the noise variance, and the constellation size $M$. To improve robustness, the SNR can alternatively be sampled from a predefined range during mini-batch training.

\section{Numerical Results}

In this section, we evaluate the proposed constellation design framework under both memory and memoryless operation. We first consider the memory regime and assess the state-adaptive \emph{information-based} constellation in terms of average SER and the resulting R-R tradeoff. We then consider steady-state operation, where the AE-based framework is evaluated for both \emph{information-based} and \emph{energy-constrained} constellation design. In the \emph{information-based} case, the learned constellations are compared with the algorithmic design to provide a reference solution. Finally, the learned \emph{energy-constrained} constellations are used to characterize the resulting R-E tradeoff.

\subsection{Simulation Setup}
\label{subsec:simulation_setup}

We consider amplitude modulation with non-negative levels and one channel use per message, \textit{i.e.,} $n=1$. Unless otherwise stated, the constellation size is selected from $M\in\{4,8,16\}$. The AE is trained separately for each considered transmit power level and EH requirement. Each training epoch consists of $M\times 10^{7}$ transmitted symbols, with batch size $M\times 10^{4}$ and $10^{3}$ steps per epoch. The networks are trained for 50 epochs using the Adam optimizer with initial learning rate $10^{-3}$. The AE architecture is summarized in Table~I. In addition, the rectifier is modeled as a Schottky-based half-wave rectifier with parameters $\eta=1.05$, $V_T=25.85~\mathrm{mV}$, $I_s=5~\mu\mathrm{A}$, $R_{\mathrm{in}}=50~\Omega$, and $R_L=8.25~\mathrm{k}\Omega$~\cite{Dual}. For the memory-aware simulations, the capacitance is set to $C=10~\mathrm{nF}$. Finally, the fitted parameters of the discrete-time memory model and the skew-normal approximation of the rectifier state distribution are reported in Table~\ref{tab:fitted_memory_params}.

	\subsection{Operation with Memory: Average SER and R-R Tradeoff}
	\label{subsec:memory_results}
	
	We now evaluate the proposed state-adaptive information-based design in the presence of rectifier memory. Unless otherwise stated, we set $M=4$ and $P_T=-10$ dBm, and vary the symbol duration over $T_s\in\{1,3,5,7,10\}\, \mu\mathrm{s}$. For the state-adaptive policy, the rectifier state distribution is obtained by simulating the transmission process and is then fitted using the skew-normal model in (18). For the fixed benchmark, a separate Monte Carlo simulation is performed, since the fixed constellation induces a different state distribution. For the affine approximation, $v_{\min}$ and $v_{\max}$ are chosen to cover the $90\%$ of the rectifier states.

		\begin{table}[t]
		\centering
		\caption{Fitted memory-aware model and skew-normal distribution parameters.}
		\label{tab:fitted_memory_params}
		\begin{tabular}{c c c c c c}
			\hline\hline
			$T_s~\!\!\!$ 
			& $\alpha_f \!\!\!$ 
			& $\beta_f \!\!\!$ 
			& $\xi\!\!\!$ 
			& $\omega\!\!\!$ 
			& $\alpha\!\!\!$ \\[1mm]
			\hline
			$1\, \mu\mathrm{s}$  
			& $0.3763$ & $0.3752$ 
			& $5.03{\times}10^{-2}$ 
			& $4.07{\times}10^{-3}$ 
			& $-0.9459$ \\[1mm]
			$3\, \mu\mathrm{s}$  
			& $0.4794$ & $0.4762$ 
			& $5.23{\times}10^{-2}$ 
			& $7.15{\times}10^{-3}$ 
			& $-1.3957$ \\[1mm]
			$4\, \mu\mathrm{s}$  
			& $0.4794$ & $0.4782$ 
			& $5.32{\times}10^{-2}$ 
			& $8.57{\times}10^{-3}$ 
			& $-1.5225$ \\[1mm]
			$5\, \mu\mathrm{s}$  
			& $0.4809$ & $0.4812$ 
			& $5.41{\times}10^{-2}$ 
			& $9.81{\times}10^{-3}$ 
			& $-1.6133$ \\[1mm]
			$7\, \mu\mathrm{s}$  
			& $0.4866$ & $0.4750$ 
			& $5.54{\times}10^{-2}$ 
			& $1.21{\times}10^{-2}$ 
			& $-1.6711$ \\[1mm]
			$10\, \mu\mathrm{s}$ 
			& $0.4427$ & $0.4844$ 
			& $5.90{\times}10^{-2}$ 
			& $1.62{\times}10^{-2}$ 
			& $-2.3184$ \\[1mm]
			\hline\hline
		\end{tabular}
	\end{table}

Fig.~6 validates the average SER characterization of the state-adaptive policy. The markers correspond to the numerical evaluation of (19), while the solid curves are obtained from the closed-form approximation in (21). The close agreement between the two results shows that the skew-normal approximation of the state distribution, together with the affine approximation of $d_{\max}(v)$, is sufficiently accurate in the considered memory regime. Therefore, (21) provides a compact and accurate way to predict the average SER without repeatedly averaging. The figure also shows the impact of $T_s$ on reliability. As $T_s$ increases, the capacitor has more time to move toward the steady-state voltage associated with each transmitted symbol. This reduces memory-induced voltage compression and increases the separation between adjacent sampled voltage levels. Consequently, the same SER can be achieved at a lower SNR, however, this improvement comes at the cost of a lower symbol rate, which motivates the R-R tradeoff studied below.

Fig.~7 compares the proposed state-adaptive design with a fixed steady-state constellation. The benchmark constellation is designed from the memoryless mapping $\bar v_L(\cdot)$ and then evaluated over the memory channel. At high SNR, the state-adaptive design provides a clear gain. This is consistent with the design criterion in Algorithm~1, which maximizes the feasible distance between adjacent sampled voltage levels for each rectifier state. In this regime, detection errors are mainly governed by the minimum separation at the decoder input, and adapting the constellation to the current capacitor voltage is therefore beneficial. At low SNR, however, the fixed steady-state constellation can perform better, since the state-adaptive policy is based on a minimum-distance criterion and is therefore mainly tailored to the moderate-to-high SNR regime. When sampling noise is dominant, the SER depends not only on the minimum adjacent spacing, but also on the full set of decision regions and on the state distribution induced by the transmitted symbols.

Fig.~8 shows the resulting R-R tradeoff under rectifier memory. The effective throughput is defined as
\begin{equation}
R_{\rm eff}=\frac{\log_2(M)}{T_s}\left(1-\overline{\rm SER}\right),
\end{equation}
where $1-\overline{\rm SER}$ represents the average reliability. Equation (26) highlights the dual role of $T_s$ on the effective throughput. A shorter $T_s$ increases the symbol rate, but gives the capacitor less time to charge or discharge, which compresses the sampled voltage levels and reduces reliability. A longer $T_s$ improves the voltage separation, but lowers the symbol rate. For the considered parameters, this rate loss is stronger than the corresponding reliability gain in both SNR regimes. Therefore, the effective throughput decreases as $T_s$ increases, showing that shorter symbol durations are preferable.

		 \begin{figure}[t] \centering \includegraphics[width=\linewidth]{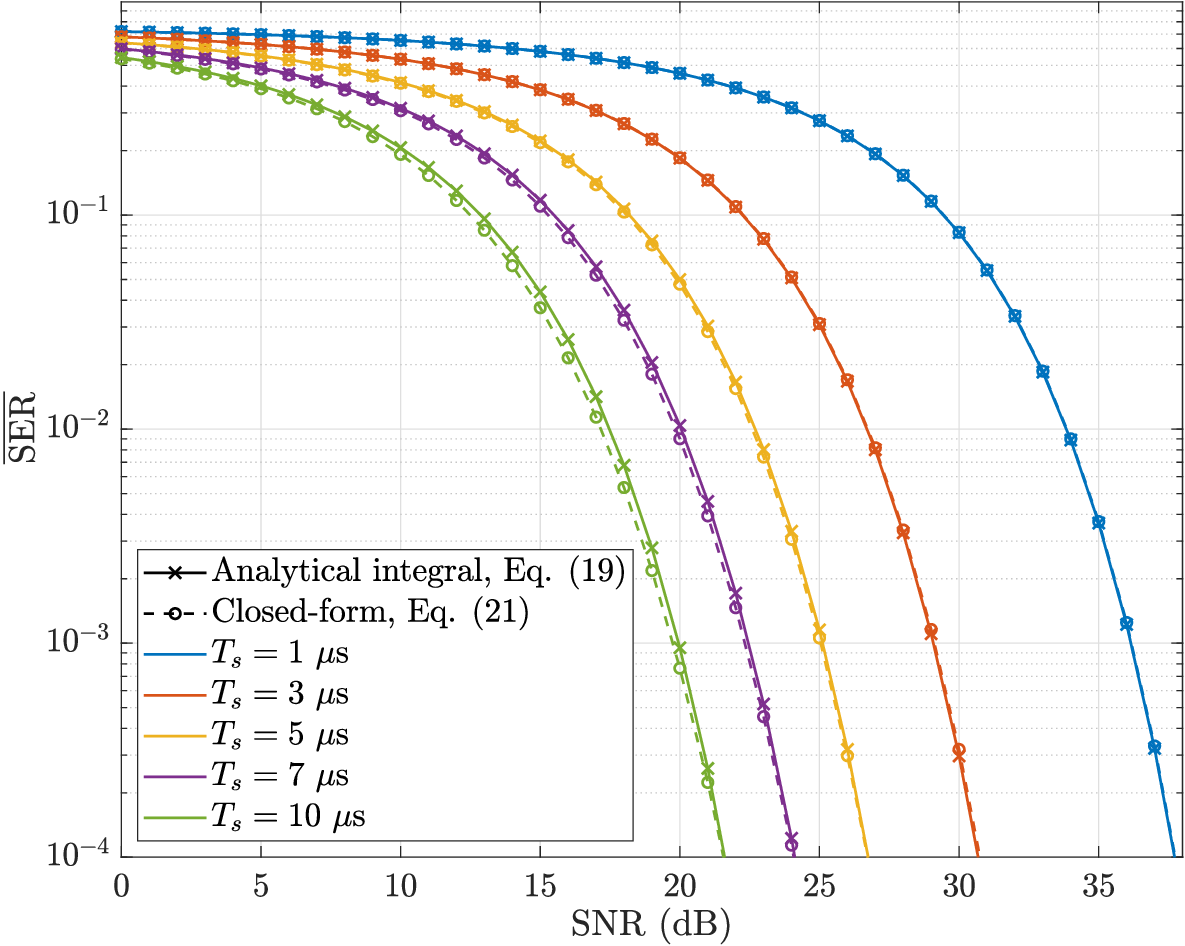}\vspace{-1mm} \caption{Closed-form $\overline{\mathrm{SER}}$ validation for the state-adaptive policy.} \label{fig:closedform_validation_sub}\vspace{-4mm} 
	\end{figure} 
	
	\begin{figure}[t] \centering
		\includegraphics[width=\linewidth]{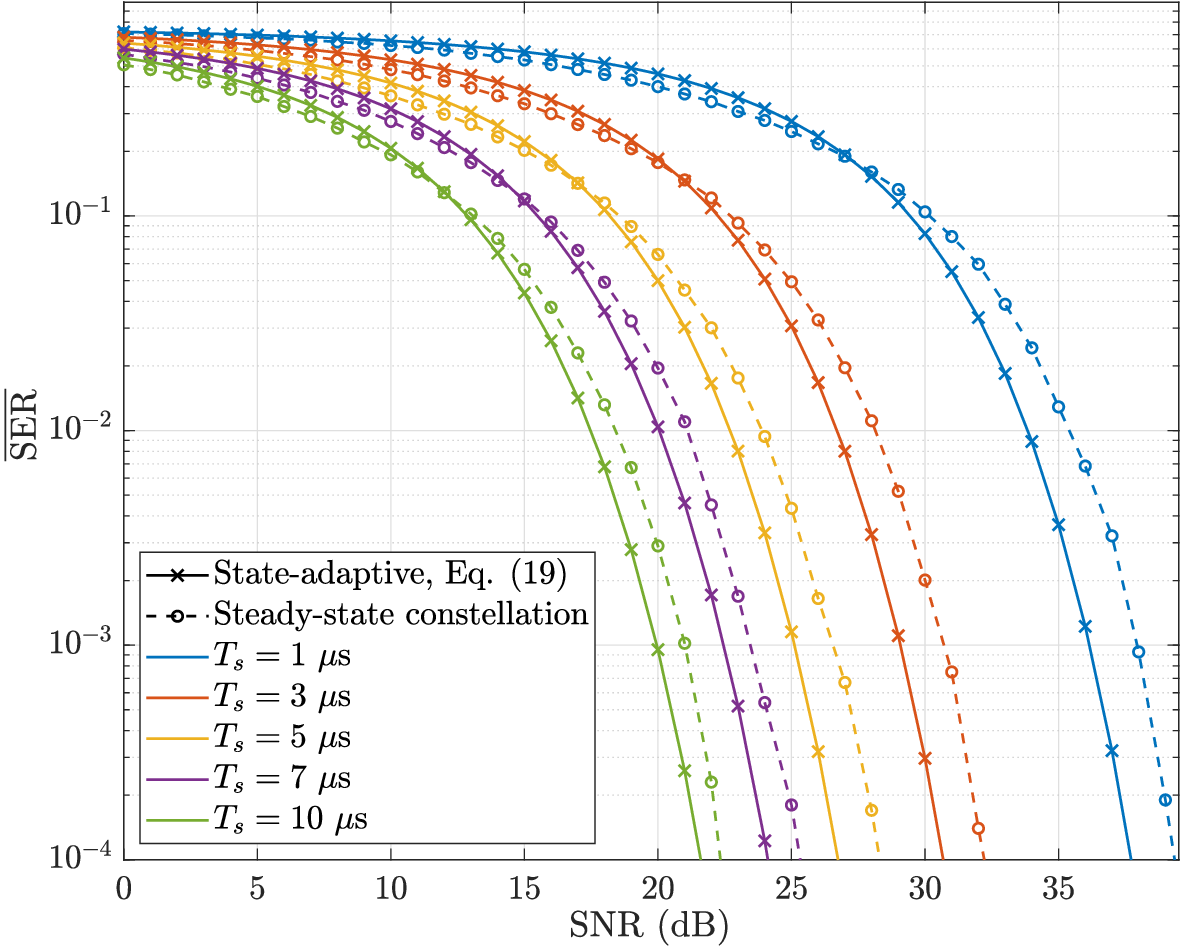}\vspace{-1mm}
		\caption{$\overline{\mathrm{SER}}$ comparison between the state-adaptive design and the fixed steady-state constellation.}
		\label{fig:ser_benchmark_memory_sub}\vspace{-4mm} 
	\end{figure}
	
	\subsection{Memoryless Operation: Information-based Constellations and SER Performance}
	\label{subsec:memoryless_info_results}
	
We now investigate the memoryless steady-state case by selecting $T_s=1~\mathrm{s}$. The learned constellations are obtained by separately training different AEs for every different $P_T$ level. In Fig.~\ref{fig:ae_constellations}(a) we plot the learned \emph{information-based} constellations ($\lambda=0$), along with the ones provided by the algorithmic approach, for $M\in\{4,8\}$ and various $P_T$ levels. It is expected that for different $P_T$, we operate in different rectification regimes, since the rectification function given by~\eqref{eq:vbar_closed_form} depends on $P_T$. We can observe that for high $P_T$ values (\textit{e.g.}, $P_T=10$ dBm), the learned constellations consist of one zero-amplitude point and $M-1$ further separated points, with equal distance between them. On the other hand, for small $P_T$ values (\textit{e.g.}, $P_T=-20$ dBm), the learned constellations for $M=\{4,8\}$ still have a zero-amplitude point, however the remaining $M-1$ points are not equally separated. In particular, the distance between consecutive points gradually decreases. The above observations result from the dependency of the nonlinear rectification characteristics with the transmit power $P_T$. In fact, for high values of $P_T$, the expression given in (4) becomes somewhat linear for high amplitude inputs, while for low $P_T$ values becomes exponentially increasing.

 \begin{figure}[t] \centering 
 	\begin{subfigure}{.493\linewidth} \centering \includegraphics[width=\linewidth]{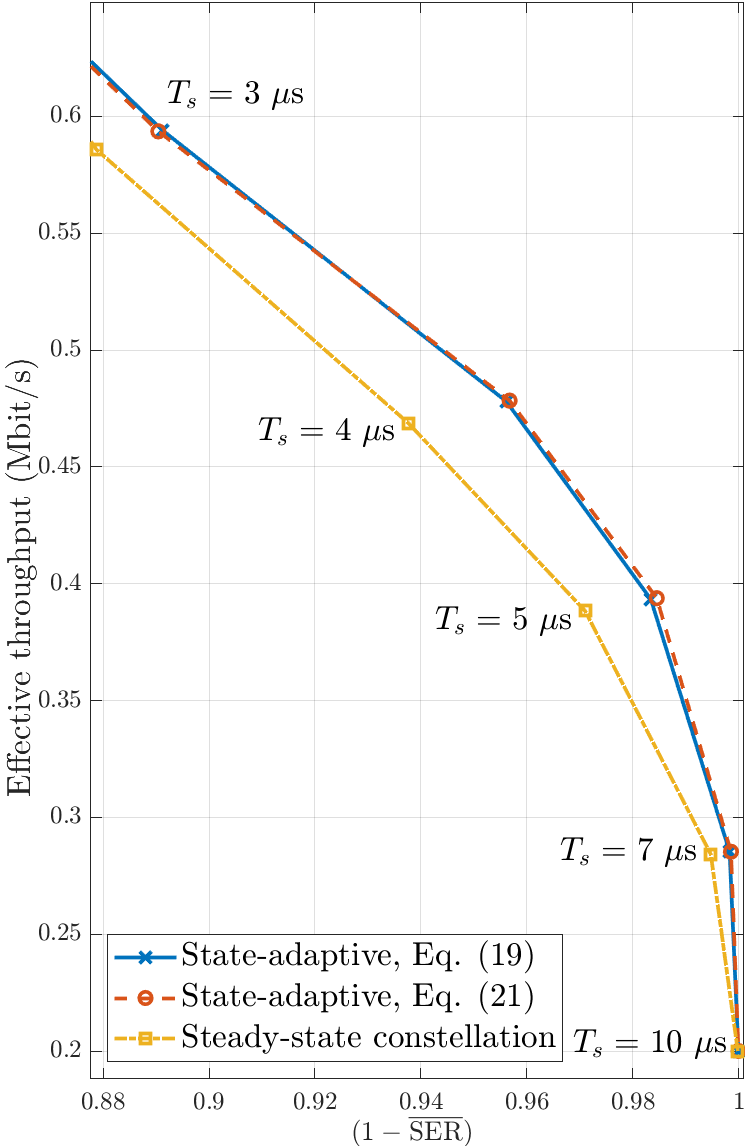}\vspace{-1mm} \caption{High SNR regime, $\sigma_s^2=-31~\mathrm{dBm}$.} \label{fig:rr_tradeoff_high_sub} 
	\end{subfigure} 
	\begin{subfigure}{.493\linewidth} \centering \includegraphics[width=\linewidth]{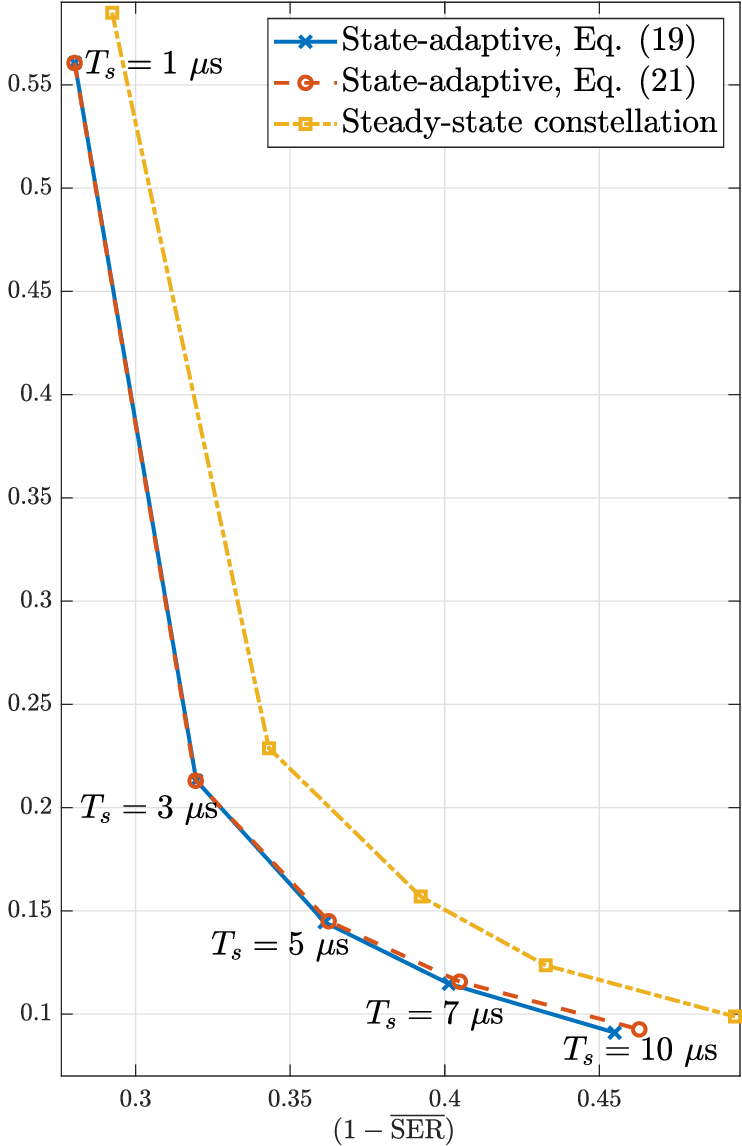}\vspace{-1mm} \caption{Low SNR regime, $\sigma_s^2=-10~\mathrm{dBm}$.} \label{fig:rr_tradeoff_low_sub} 
	\end{subfigure} 
	\caption{R-R tradeoff under rectifier memory for different symbol durations.} \label{fig:rr_tradeoff_memory}
	\vspace{-4mm} 
\end{figure} 

\begin{figure*}[t] \centering \begin{subfigure}{.49\textwidth} \centering \includegraphics[width=\linewidth]{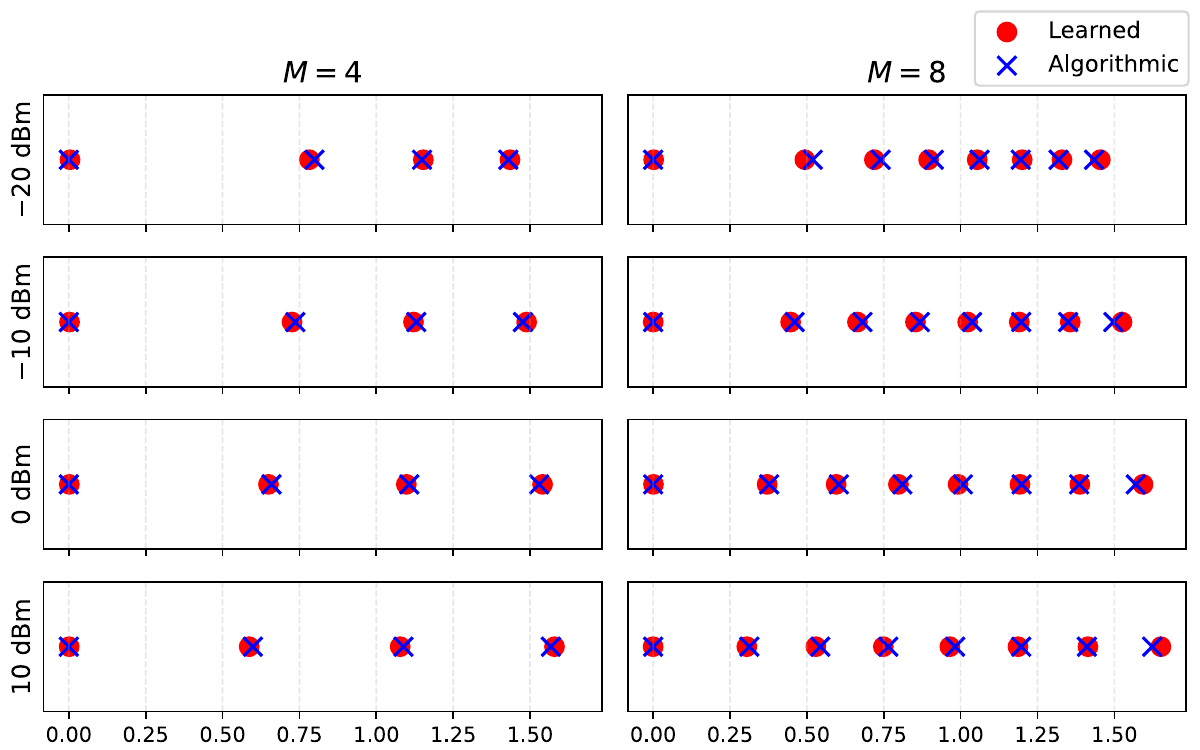} \caption{Information-based constellations.} \label{fig:ae_constellations_info} 
\end{subfigure}
\begin{subfigure}{.49\textwidth} \centering \includegraphics[width=\linewidth]{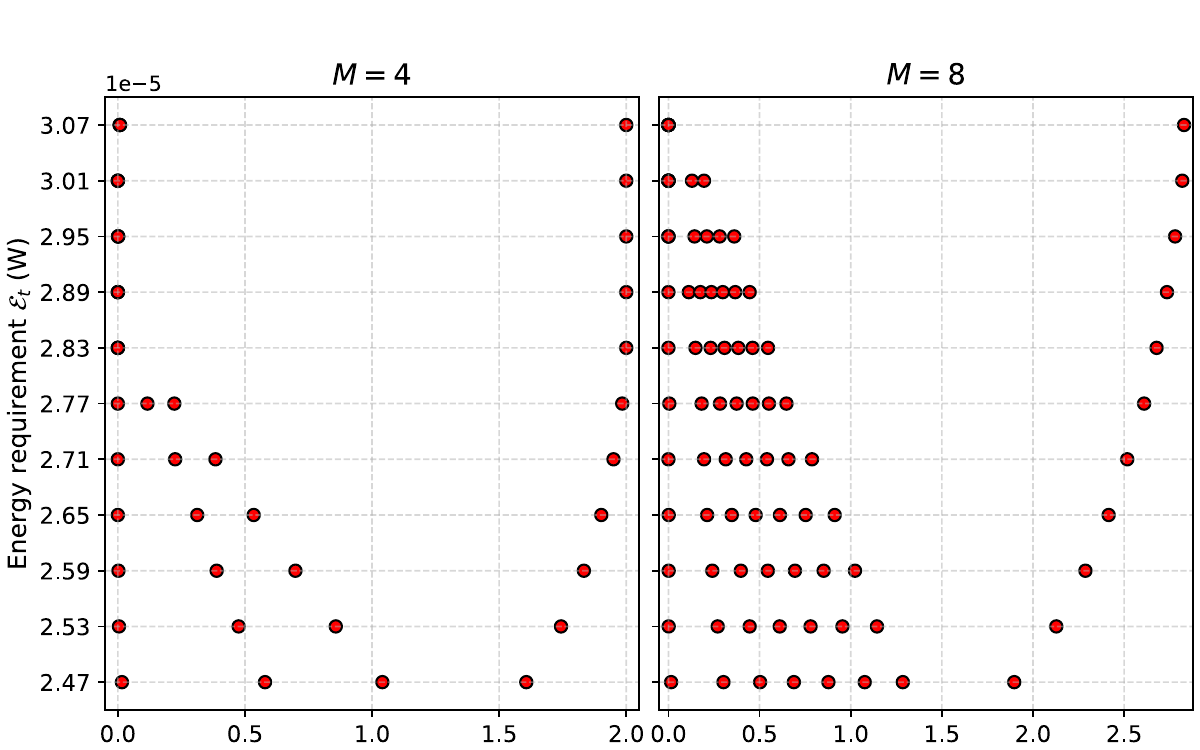} \caption{Energy-constrained constellations.} \label{fig:ae_constellations_energy} \end{subfigure}\vspace{-1mm} \caption{Learned constellations obtained with the proposed AE-based framework. (a) Information-based constellations ($\lambda=0$) for different transmit power levels $P_T$ and constellation sizes $M\in\{4,8\}$, compared with the algorithmic construction. (b) Energy-constrained constellations learned for $M\in\{4,8\}$ under different EH requirements $\mathcal{E}_t$, at fixed transmit power $P_T=5~\mathrm{dBm}$.} \label{fig:ae_constellations}\vspace{-4mm} 
\end{figure*}

 \begin{figure}[t] \centering \includegraphics[width=\linewidth]{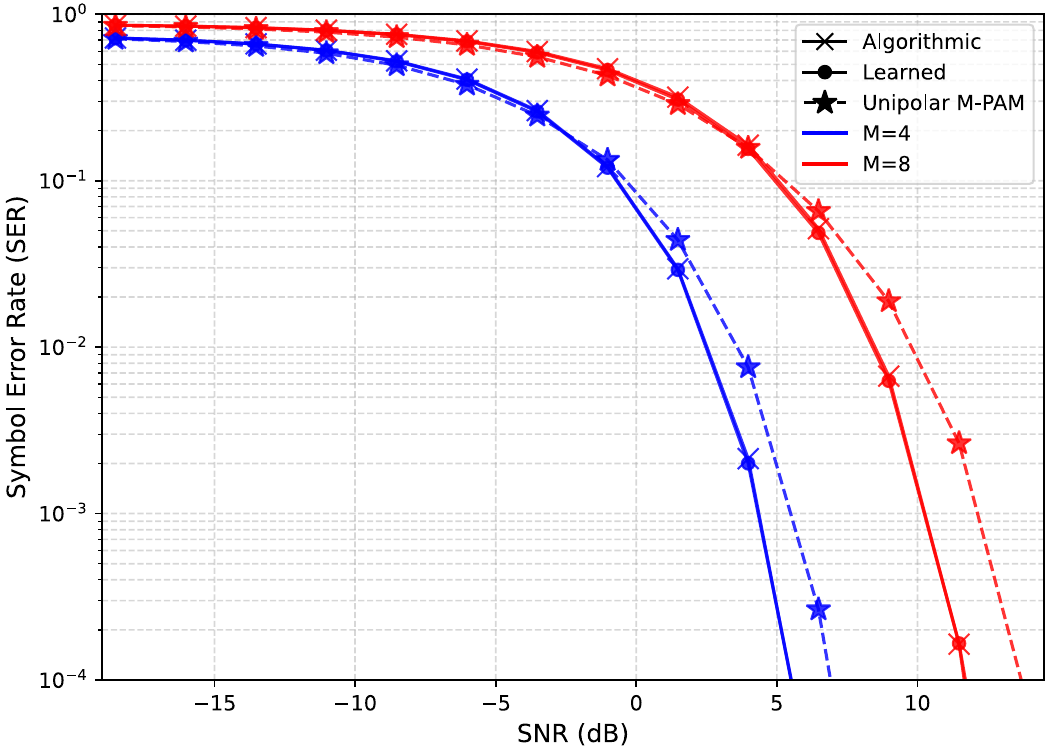}\vspace{-1mm} \caption{SER performance of the learned and algorithmic information-based constellations, compared against unipolar $M$-PAM for different transmit power levels.} \label{fig:ser_memoryless_info}\vspace{-4mm} 
\end{figure} 

In Fig.~\ref{fig:ser_memoryless_info} we plot the SER curves with respect to $\text{SNR}=\frac{P_T}{\sigma_s^2}$, under a fixed noise variance $\sigma_s^2= -4$ dBm, for $M\in\{4,8\}$. Both the learned and the algorithmic \emph{information-based} constellations are compared against unipolar pulse amplitude modulation (PAM), where $x_i=i \Delta$, with $i\in \{0,1,\ldots,M-1\}$ and $\Delta=\sqrt{\frac{6}{\left(M-1\right)\left(2M-1\right)}}$. The gains provided by the learned constellations can be realized in the high SNR regime, where they consistently outperform unipolar $M$-PAM under the same UR-SWIPT architecture, improving up to an order-of-magnitude the SER performance. As expected, the minimum Euclidean distance is reduced as $M$ increases, thereby degrading SER performance. Overall, the efficiency of our proposed learning approach can be confirmed, since the learned \emph{information-based} constellations and their corresponding SER curves match with those provided by the algorithmic approach.

 \subsection{Memoryless Operation: Energy-constrained Constellations and R-E Tradeoff}
 \label{subsec:memoryless_energy_results} 
 
 We now examine the R-E tradeoff in the presence of an energy requirement \(\mathcal{E}_{t}\) and study how different EH constraints shape the resulting \emph{energy-constrained} constellations. To this end, we train the AE using the loss function in (24), with sufficiently high $\lambda$. In Fig.~\ref{fig:ae_constellations}(b), we plot the learned constellations for various \(\mathcal{E}_{t}\) with fixed \(P_T=5\) $\mathrm{dBm}$ and $\lambda$ varying from \(10^{12}\) to \(10^{14}\). We note that \(\lambda\) depends on \(\mathcal{E}_{t}\), the constellation cardinality \(M\), and the transmit power \(P_T\). Each constellation is obtained by training a separate AE for the specified \(\mathcal{E}_{t}\), starting from a minimum value \(\mathcal{E}_{\min}\) corresponding to the harvested energy provided by the \emph{information-based} constellation. We can observe that by increasing the energy requirements, the resulting \emph{energy-constrained} constellations are characterized by a higher PAPR, until they finally converge to OOK-like signaling which is known to be optimal for EH purposes (low probability on ON and high probability on OFF \cite{Morsi}). Interestingly, the learned constellations tend to concentrate power into a single “power-symbol” while reducing the amplitudes of the remaining non-zero symbols to satisfy the average-power constraint. This behavior can be explained by the rectifier's exponential trend, where the presence of a “power-symbol” enables the system to meet the EH requirements while allowing for reliable decoding.

  \begin{figure}[t] 
 	\centering \includegraphics[width=\linewidth]{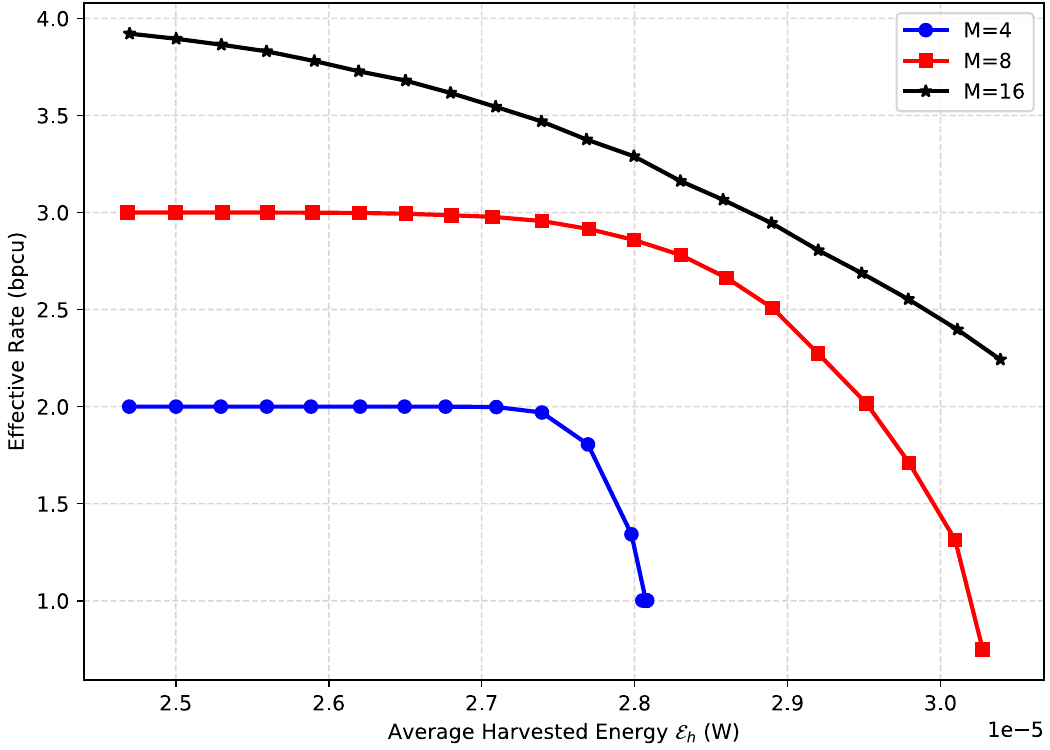}\vspace{-1mm} \caption{Effective R-E tradeoff achieved by the learned energy-constrained constellations for different EH requirements, with $M\in\{4,8,16\}$ and $P_T=5~\mathrm{dBm}$.} \label{fig:re_tradeoff_memoryless}\vspace{-4mm} 
 \end{figure}
 
 Fig.~\ref{fig:re_tradeoff_memoryless} illustrates the tradeoff between the effective rate \(R=\log_2(M)\,(1-\mathrm{SER})\) in bits per channel use (bpcu) and the average harvested energy \(\mathcal{E}_{h}\). We sweep the minimum EH requirement \(\mathcal{E}_t\) and the constellation size \(M\in\{4,8,16\}\) under a fixed transmit power \(P_T=5\) $\mathrm{dBm}$ and noise variance \(\sigma_s^2=-10\) $\mathrm{dBm}$. For $M = 4$ and $M = 8$, we observe that for several EH constraints, the system achieves nearly perfect ID accuracy, enabling reliable decoding while simultaneously satisfying the EH requirements. However, as $M$ increases, decoding reliability decreases due to higher error rates, but the effective rate improves. This tradeoff arises because larger constellations yield higher peak amplitudes (i.e., higher PAPR), which benefits EH due to the rectifier's exponential response, while also providing more information bits per channel use. Overall, increasing $M$ shifts the tradeoff toward higher rates but lower reliability. This effect is mainly due to the high-SNR regime from the relatively large $P_T$, and reducing $P_T$ would alter the tradeoff.

\section{Conclusion}

In this paper, we developed a constellation design framework for nonlinear UR-SWIPT receivers with channel memory. We first proposed a tractable discrete-time rectification model that captures both the nonlinear steady-state diode response and the asymmetric capacitor charging/discharging dynamics. Based on this model, we characterized the resulting state-dependent information channel and showed how rectifier memory reshapes the effective constellation at the decoder. For the \emph{information-based} setting, we derived algorithmic constellation designs that enforce uniformly spaced sampled voltage levels and obtained an average SER characterization by accounting for the stationary distribution of the rectifier state. Numerical results validated the proposed SER approximation and highlighted the role of symbol duration in the R-R tradeoff. For the \emph{energy-constrained} setting, we proposed an AE-based framework that embeds the rectification model as a differentiable channel block, enabling end-to-end optimization of the transmit constellation and detection rule under EH constraints. The learned constellations adapt to the nonlinear rectifier response and develop high-amplitude power symbols when the EH requirement becomes stringent. Overall, the results demonstrate that accounting for rectifier nonlinearity and memory is essential for reliable and energy-efficient UR-SWIPT design. Future work may consider sequence-level AE training, robust designs without instantaneous state information, and multitone or coded modulation schemes for nonlinear rectification channels with memory.
	
	\appendices
	\section{Discrete-Time Memory-Aware Rectifier Model}
	\label{app:memory_model}
	
	During the $n$-th symbol interval, $t\in[(n-1)T_s,nT_s)$, the transmit amplitude is fixed to $A^{(n)}$. Since $v_{\mathrm{in}}(t)$ varies at the carrier frequency whereas $v_L(t)$ varies slowly for sufficiently large $C$, we first average the diode current over one RF period. Specifically, we define
\begin{equation}
	I_{d,\mathrm{avg}}(A,v)
	\triangleq
	\frac{1}{T}\int_0^T
	I_s\left(
	e^{\frac{v_{\mathrm{in}}(t)-v}{\eta V_T}}-1
	\right)\mathrm{d}t,
	\label{eq:Idavg_def}
\end{equation}
where $T=1/f_c$. The slow envelope dynamics over the $n$-th symbol interval can then be approximated
\begin{equation}
	C\frac{\mathrm{d}v_L(t)}{\mathrm{d}t}
	=
	I_{d,\mathrm{avg}}\big(A^{(n)},v_L(t)\big)
	-
	\frac{v_L(t)}{R_L}.
	\label{eq:avg_envelope}
\end{equation}

Let $\bar v_L(A^{(n)})$ denote the steady-state voltage corresponding to $A^{(n)}$, and define the local deviation as
$\delta v(t)\triangleq v_L(t)-\bar v_L(A^{(n)})$. Linearizing the averaged diode current around $\bar v_L(A^{(n)})$ gives
\begin{equation}
	I_{d,\mathrm{avg}}\big(A^{(n)},v_L(t)\big)\!
	\approx\!
	I_{d,\mathrm{avg}}\big(A^{(n)},\bar v_L(A^{(n)})\big)
	-
	g_d(A^{(n)})\delta v(t),
	\label{eq:Idavg_linearized}
\end{equation}
where
\begin{equation}
	g_d(A^{(n)})
	\triangleq
	-
	\left.
	\frac{\partial I_{d,\mathrm{avg}}(A^{(n)},v)}
	{\partial v}
	\right|_{v=\bar v_L(A^{(n)})}
	\ge 0
	\label{eq:gd_def}
\end{equation}
is the incremental diode conductance. At steady-state, the capacitor current is zero, so
$I_{d,\mathrm{avg}}\big(A^{(n)},\bar v_L(A^{(n)})\big)=\bar v_L(A^{(n)})/R_L$.
Substituting this relation and~\eqref{eq:Idavg_linearized} into~\eqref{eq:avg_envelope} yields the first-order relaxation
\begin{equation}
	\frac{\mathrm{d}\delta v(t)}{\mathrm{d}t}
	=
	-\frac{1}{\tau_{\mathrm{eff}}(A^{(n)})}\delta v(t),
	\label{eq:delta_relaxation}
\end{equation}
where\vspace{-1mm}
\begin{equation}
	\tau_{\mathrm{eff}}(A^{(n)})
	=
	\frac{C}
	{g_d(A^{(n)})+\frac{1}{R_L}}.
	\label{eq:tau_eff_general}\vspace{-1mm}
\end{equation}
Solving~\eqref{eq:delta_relaxation} over one symbol interval gives
	\begin{equation} 
	v_L(t) \!=\! \bar v_L(A^{(n)}) + \bigl(v_L((n-1)T_s)-\bar v_L(A^{(n)})\bigr) e^{\left(\! -\frac{t-(n-1)T_s}{\tau_{\mathrm{eff}}(A^{(n)})}\! \right)}. 
	\label{eq:ct_relaxation_solution}
\end{equation}
Sampling~\eqref{eq:ct_relaxation_solution} at $t=nT_s$ gives the discrete-time model in~\eqref{eq:dt_model}.

It remains to specify $\tau_{\mathrm{eff}}(A^{(n)})$ in the charging and discharging regimes. When
$v_L^{(n-1)}<\bar v_L(A^{(n)})$, the diode conducts and the incremental conductance in~\eqref{eq:gd_def} is non-negligible. Since $I_{d,\mathrm{avg}}(A^{(n)},v)$ depends on $v$ through the factor $e^{-v/(\eta V_T)}$, and assuming $\bar v_L(A^{(n)})/R_L\gg I_s$, we obtain
\begin{equation}
	g_d(A^{(n)})
	\approx
	\frac{I_{d,\mathrm{avg}}\big(A^{(n)},\bar v_L(A^{(n)})\big)}
	{\eta V_T}
	\approx
	\frac{\bar v_L(A^{(n)})}{\eta V_T R_L}.
	\label{eq:gd_surrogate}
\end{equation}
Substituting~\eqref{eq:gd_surrogate} into~\eqref{eq:tau_eff_general} yields the charging time constant
\begin{equation}
	\tau_{\mathrm{eff}}(A^{(n)})
	\approx
	\frac{R_L C}
	{1+\bar v_L(A^{(n)})/(\eta V_T)}.
	\label{eq:tau_gd_app}
\end{equation}
	Thus, larger rectified voltages correspond to stronger diode conduction and faster charging.
	
Conversely, when $v_L^{(n-1)}\ge \bar v_L(A^{(n)})$, the diode is mostly reverse-biased and its incremental conductance is negligible, \text{i.e.}, $g_d(A^{(n)})\approx0$. The capacitor then discharges mainly through the load resistance, leading to the passive RC time constant\vspace{-1mm}
	\begin{equation}
		\tau_{\mathrm{eff}}(A^{(n)})
		\approx
		R_LC.
		\label{eq:tau_discharge_app}\vspace{-1mm}
	\end{equation}
	To account for the mismatch caused by the above approximations, we introduce fitted parameters $\alpha_f$ and $\beta_f$, calibrated by matching the discrete-time recursion to nonlinear ODE simulations over representative symbol sequences. This gives the piecewise approximation in (\ref{eq:tau_piecewise}).

	\section{Closed-form SER Approximation}
	\label{app:ser_skew_normal}
	
	Starting from \eqref{eq:avg_ser_skew_normal}, we use the affine approximation $d_{\max}(v)\approx d_0+d_1v$ over the high-probability range of the rectifier state. Let
	\begin{equation}
		z=\frac{v-\xi}{\omega},\qquad
		A=\frac{d_0+d_1\xi}{2\sigma_s},\qquad
		B=\frac{d_1\omega}{2\sigma_s},
		\label{eq:app_AB_def}
	\end{equation}
	so that $Q\!\left(\frac{d_{\max}(v)}{2\sigma_s}\right)\approx Q(A+Bz)$. Therefore, the SER follows
	\begin{equation}
		\overline{\mathrm{SER}}
		\approx
		\frac{4(M-1)}{M}
		\int_{-\infty}^{\infty}
		\phi(z)\Phi(\alpha z)\,Q(A+Bz)\,\mathrm{d}z.
		\label{eq:app_avg_ser_standardized}
	\end{equation}
	
	Let $
		I \triangleq \int_{-\infty}^{\infty}
		\phi(z)\Phi(\alpha z)\,Q(A+Bz)\,\mathrm{d}z ,$ and $G,U,W$ be mutually independent standard normal random variables. Using
	$\Phi(\alpha g)=\Pr(U\le \alpha g)$ and
	$Q(A+Bg)=\Pr(W\ge A+Bg)$, we obtain
	\begin{equation}
		\begin{aligned}
			I
			&=
			\int_{-\infty}^{\infty}
			\phi(g)\Pr(U\le \alpha g)\Pr(W\ge A+Bg)\,\mathrm{d}g\\
			&=
			\Pr\!\left(U\le \alpha G,\; W\ge A+BG\right),
		\end{aligned}
		\label{eq:app_I_prob}
	\end{equation}
	where the second equality follows by averaging over $G\sim\mathcal{N}(0,1)$. Equivalently,
	\begin{equation}
		I
		=
		\Pr\!\left(U-\alpha G\le 0,\; BG-W\le -A\right).
		\label{eq:app_I_prob_equiv}
	\end{equation}
	Now define
	\begin{equation}
		X_1
		=
		\frac{U-\alpha G}{\sqrt{1+\alpha^2}},
		\qquad
		X_2
		=
		\frac{BG-W}{\sqrt{1+B^2}}.
		\label{eq:app_X_def}
	\end{equation}
	Since $G$, $U$, and $W$ are independent standard normal variables, $(X_1,X_2)$ is a zero-mean bivariate Gaussian vector with unit variances. Its correlation coefficient is
	\begin{equation}
		\rho
		=
		\mathbb{E}[X_1X_2]
		=
		-\frac{\alpha B}
		{\sqrt{(1+\alpha^2)(1+B^2)}}.
		\label{eq:app_rho_def}
	\end{equation}
	Therefore,
	\begin{equation}
		I
		=
		\Phi_2\!\left(0,\;-\frac{A}{\sqrt{1+B^2}};\rho\right),
		\label{eq:app_I_closed}
	\end{equation}
	and substituting \eqref{eq:app_I_closed} into \eqref{eq:app_avg_ser_standardized} yields \eqref{eq:avg_ser_closed_form_expanded}.

\end{document}